\documentclass[a4paper,11pt]{article}

\usepackage{tikzexternal}
\def\tornpaper#1{%
\tikz{
  \node[inner sep=1em] (A) {#1};  
  \begin{pgfonlayer}{background}  
  \fill[paper] 
     \pgfextra{\pgfmathsetseed{\arabic{mathseed}}\addtocounter{mathseed}{1}}%
      {decorate[irregular cloudy border]{decorate{decorate{decorate{decorate[ragged border]{
        (A.north west) -- (A.north east)
      }}}}}}
      -- (A.south east)
     \pgfextra{\pgfmathsetseed{\arabic{mathseed}}}%
      {decorate[irregular spiky border]{decorate{decorate{decorate{decorate[ragged border]{
      -- (A.south west)
      }}}}}}
      -- (A.north west);
  \end{pgfonlayer}}
}

\usepackage{graphicx} 	
\usepackage{subfig} 
\usepackage{float}

\graphicspath{{./Images/}}

\usepackage[ pdftex, plainpages = false, pdfpagelabels,
           bookmarks=false,
           bookmarksopen = true,
           bookmarksnumbered = true,
           breaklinks = true,
           linktocpage,
           pagebackref,
           colorlinks = true, 
           linkcolor = blue,
           urlcolor = blue,
           citecolor = red,
           anchorcolor = green,
           hyperindex = true,
           hyperfigures
           ]{hyperref}

\usepackage{amsthm}
\usepackage{amsmath}
\usepackage{amssymb}
\usepackage{amsfonts}
\usepackage{amstext}
\usepackage{xspace}
\usepackage{graphicx}
\usepackage{mathrsfs}
\usepackage[shortlabels]{enumitem}
\setlist{noitemsep}

\usepackage{vmargin}
\setmarginsrb{1in}{1in}{1in}{1in}{0pt}{0pt}{0pt}{6mm}

\newcommand{\yes}{\textsc{yes}\xspace}

\newtheorem*{untheorem}{Theorem}
\newtheorem{theorem}{Theorem}
\newtheorem{lemma}{Lemma}[section]
\newtheorem{claim}{Claim}[section]

\newenvironment{claimproof}{\begin{proof}\renewcommand{\qedsymbol}{\claimqed}}{\end{proof}\renewcommand{\qedsymbol}{\plainqed}}

\let\plainqed\qedsymbol

\usepackage{boxedminipage}

\newcommand{\OhTilde}{{\widetilde{\mathcal{O}}}}
\newcommand{\Oh}{{\mathcal{O}}}

\newcommand{\cnfsat}{\textsc{cnf-sat}\xspace}
\newcommand{\kOPT}{\textsc{$k$-opt}\xspace}
\newcommand{\opt}{\textsc{opt}\xspace}
\newcommand{\kOPTOptimization}{\textsc{$k$-opt Optimization}\xspace}
\newcommand{\kOPTDetection}{\textsc{$k$-opt Detection}\xspace}
\newcommand{\TwoOPT}{\textsc{$2$-opt}\xspace}
\newcommand{\TwoOPTDetection}{\textsc{$2$-opt Detection}\xspace}
\newcommand{\ThreeOPT}{\textsc{$3$-opt}\xspace}
\newcommand{\ThreeOPTImprovement}{\textsc{$3$-opt Detection}\xspace}
\newcommand{\ThreeOPTDetection}{\textsc{$3$-opt Detection}\xspace}

\newcommand{\TSP}{\textsc{tsp}\xspace}
\newcommand{\APSP}{\textsc{apsp}\xspace}
\newcommand{\NegativeTriangle}{\textsc{Negative Edge-Weighted Triangle}\xspace}

\newcommand{\defproblem}[3]{
 \vspace{1mm}
\noindent\fbox{
 \begin{minipage}{0.96\textwidth}
 \begin{tabular*}{\textwidth}{@{\extracolsep{\fill}}lr} #1 &  \\ \end{tabular*}
 {\bf{Input:}} #2 \\
 {\bf{Question:}} #3
 \end{minipage}
 }
 \vspace{1mm}
}

\usepackage{bold-extra}
\usepackage{lmodern}
\usepackage[T1]{fontenc}
\usepackage{textcomp}

\newcommand{\Reals}{{\mathbb{R}}}

\DeclareMathOperator{\polylog}{polylog}
\newcommand{\bd}{\partial}

\newcommand{\dist}{\mathrm{dist}}

\newcommand{\true}{\mbox{\sc True}}
\newcommand{\false}{\mbox{\sc False}}

\newcommand{\tree}{\ensuremath{\mathcal{T}}}
\newcommand{\node}{\nu}

\newcommand{\myroot}{\mathit{root}}

\newcommand{\mylist}{\ensuremath{\mathcal{L}}}
\newcommand{\ds}{\ensuremath{\mathcal{DS}}}

\newcommand{\mypath}{P}
\newcommand{\rev}{\mbox{\sc Rev}}
\newcommand{\mc}{\mathit{MinCost}}
\newcommand{\mrc}{\mathit{MinRevCost}}
\newcommand{\myrev}{\mathit{rev}}

\newcommand{\D}{\ensuremath{\mathcal{D}}}

\newcommand{\floor}[1]{\left\lfloor #1 \right\rfloor}

\newcommand{\eps}{\varepsilon}

\newcommand{\etal}{\emph{et al.}\xspace}

\newcommand{\tsp}{{\sc tsp}\xspace}

\newcommand{\mytop}{\mathrm{top}}
\newcommand{\mybot}{\mathrm{bot}}
\newcommand{\myleft}{\mathrm{left}}
\newcommand{\myright}{\mathrm{right}}
\newcommand{\union}{\mathrm{Union}}
\newcommand{\awvd}{\mathrm{AWVD}}
\newcommand{\vd}{\mathrm{VD}}

\newcommand{\TwoOptOPT}{\textsc{2-opt Optimization}\xspace}
\newcommand{\TwoOptDET}{\textsc{2-opt Detection}\xspace}

\newcommand{\ThreeOptDET}{\textsc{3-opt Detection}\xspace}
\newcommand{\kmove}{$k$-move\xspace}

\newcommand{\Twomove}{$2$-move\xspace}
\newcommand{\Twomoves}{$2$-moves\xspace}
\newcommand{\Threemove}{$3$-move\xspace}
\newcommand{\Threemoves}{$3$-moves\xspace}

\newcommand{\BeginMyItemize}{\begin{itemize}\setlength{\itemsep}{-\parskip}}
\newcommand{\EndMyItemize}{\end{itemize}}
\newcommand{\myitemize}[1]{\BeginMyItemize #1 \EndMyItemize}

\newcommand{\BeginMyEnumerate}{\begin{enumerate}\setlength{\itemsep}{-\parskip}}
\newcommand{\EndMyEnumerate}{\end{enumerate}}

\title{Fine-Grained Complexity Analysis of Two Classic TSP Variants}

\author{Mark de Berg\thanks{Department of Mathematics and Computer Science, Eindhoven University of Technology, The Netherlands. \texttt{m.t.d.berg@tue.nl}. Supported by NWO Gravity grant ``Networks'' (024.002.003).}
\and Kevin Buchin\thanks{Department of Mathematics and Computer Science, Eindhoven University of Technology, The Netherlands. \texttt{k.a.buchin@tue.nl}. Supported by NWO grant ``A framework for progressive, user-steered algorithms in visual analytics'' (612.001.207).}
\and Bart M. P. Jansen\thanks{Department of Mathematics and Computer Science, Eindhoven University of Technology, The Netherlands. \texttt{b.m.p.jansen@tue.nl}. Supported by NWO Veni grant ``Frontiers in Parameterized Preprocessing'' and NWO Gravity grant ``Networks'' (024.002.003).}
\and Gerhard Woeginger\thanks{Department of Mathematics and Computer Science, Eindhoven University of Technology, The Netherlands. \texttt{gwoegi@win.tue.nl}. Supported by NWO Gravity grant ``Networks'' (024.002.003).}
}

\date{}
\begin{document}
\maketitle

\thispagestyle{empty}

\begin{abstract}
We analyze two classic variants of the \textsc{Traveling Salesman Problem} using the toolkit of fine-grained complexity.

Our first set of results is motivated by the \textsc{Bitonic tsp} problem: given a set of~$n$ points in the plane, compute a shortest tour consisting of two monotone chains. It is a classic dynamic-programming exercise to solve this problem in~$\Oh(n^2)$ time. While the near-quadratic dependency of similar dynamic programs for \textsc{Longest Common Subsequence} and \textsc{Discrete Fr\'echet Distance} has recently been proven to be essentially optimal under the Strong Exponential Time Hypothesis, we show that bitonic tours can be found in subquadratic time. More precisely, we present an algorithm that solves bitonic \TSP in $\Oh(n\log^2 n)$ time and its bottleneck version in $\Oh(n\log^3 n)$ time. In the more general pyramidal \TSP problem, the points to be visited are labeled $1,\ldots,n$ and the sequence of labels in the solution is required to have at most one local maximum.  Our algorithms for the bitonic (bottleneck) \TSP problem also work for the pyramidal \TSP problem in the plane. 

Our second set of results concerns the popular \kOPT heuristic for \TSP in the graph setting. More precisely, we study the \kOPT decision problem, which asks whether a given tour can be improved by a \kOPT move that replaces $k$ edges in the tour
by $k$ new edges. A simple algorithm solves \kOPT in~$\Oh(n^k)$ time for fixed~$k$. For \textsc{$2$-opt}, this is easily seen to be optimal. For~$k=3$ we prove that an algorithm with a runtime of the form~$\OhTilde(n^{3-\eps})$ exists if and only if \textsc{All-Pairs Shortest Paths} in weighted digraphs has such an algorithm. For general \textsc{$k$-opt}, it is known that a runtime of $f(k) \cdot n^{o(k/\log k)}$ would contradict the Exponential Time Hypothesis. The results for $k=2,3$ may suggest that the actual time complexity of \textsc{$k$-opt} is $\Theta(n^k)$. We show that this is not the case, by presenting an algorithm that finds the best \kmove in~$\Oh(n^{\floor{2k/3} + 1})$ time for fixed~$k \geq 3$. This implies that \textsc{$4$-opt} can be solved in~$\Oh(n^3)$ time, matching the best-known algorithm for \textsc{$3$-opt}. Finally, we show how to beat the quadratic 
barrier for $k=2$ in two important settings, namely for points in the plane and when we want to solve \TwoOPT 
repeatedly.
\end{abstract}

\newpage

\setcounter{page}{1}

\section{Introduction}

\subsection{Motivation}
We analyze two classic variants of the \textsc{Traveling Salesman Problem} (\TSP) by applying the modern
toolkit of fine-grained complexity analysis.
The first {\TSP} variant can for instance be found in Chapter~15 of the well-known textbook
\emph{``Introduction to Algorithms''} by Cormen, Leiserson, Rivest, and Stein~\cite{CormenLRS01}.
The chapter discusses dynamic programming, and its problem section poses the following classic exercise:

\begin{center}
\noindent\tornpaper{
\parbox{.93\textwidth}{\small
\textbf{\emph{15-3 Bitonic euclidean traveling-salesman problem}}\\
In the \textbf{\emph{euclidean traveling-salesman problem}}, we are given a set of~$n$ points in the plane, and
we wish to find the shortest closed tour that connects all~$n$ points. The general problem is NP-complete, and
its solution is therefore believed to require more than polynomial time.
J.~L.~Bentley has suggested that we simplify the problem by restricting our attention to \textbf{\emph{bitonic tours}},
that is, tours that start at the leftmost point, go strictly rightward to the rightmost point, and then go strictly
leftward back to the starting point. In this case, a polynomial-time algorithm is possible.
Describe an~$\Oh(n^2)$-time algorithm for determining an optimal bitonic tour.}}
\end{center}

This exercise already showed up in the very first edition of the book in 1991.
Since then, thousands of students pondered about it and (hopefully) found the solution.
One might wonder whether $\Oh(n^2)$ runtime is best possible for this problem.
As one of our main contributions, we will show that in fact it is not.

The second {\TSP} variant concerns \kOPT, a popular local search heuristic that
attempts to improve a suboptimal solution by a \emph{\kOPT move} (or: \emph{\kmove} for short),
which is an operation that removes~$k$ edges from the current
tour and reconnects the resulting pieces into a new tour by inserting~$k$ new edges.
The cases~$k=2$~\cite{Croes58} and~$k=3$ have been studied extensively with respect to various aspects such as
experimental performance~\cite{Bentley90,JohnsonM02,Lin65}, (smoothed) approximation ratio~\cite{ChandraKT99,KunnemannM15},
rate of convergence~\cite{ChandraKT99,EnglertRV14}, and algorithm engineering~\cite{FredmanJMO95,Glover96,MavroidisPP07,ONeilB15}.
The decision problem associated with \kOPT asks, given a tour in an edge-weighted graph, whether it is possible to
obtain a tour of smaller weight by replacing~$k$ edges.
There are $\Theta(n^k)$ possibilities to choose~$k$ edges that leave the current tour, and for each choice the
number of ways to reconnect the resulting pieces back into a tour is constant (for fixed~$k$).
As the weight change for each reconnection pattern can be evaluated in~$\Oh(k)$ time, this simple algorithm finds the
best \kOPT improvement in time~$\Oh(n^k)$ for each fixed~$k$.
The survey chapter~\cite{JohnsonMc97} by Johnson and McGeoch extensively discusses \kOPT.
On page~{233} they write:

\begin{center}
\noindent\tornpaper{
\parbox{.93\textwidth}{\small
To complete our discussion of running times, we need to consider the time per move as well as the number of moves.
This includes the time needed to \emph{find} an improving move (or verify that none exists), together with the time
needed to \emph{perform} the move.
In the worst case, $2$-opt and $3$-opt require $\Omega(n^2)$ and $\Omega(n^3)$ time respectively to verify local
optimality, assuming all possible moves must be considered.
}}
\end{center}

The two lower bounds in the last sentence are stated without further justification.
It is clear that finding an improving \kmove takes $\Omega(n^k)$ time,
if we require that all possible moves must be enumerated \emph{explicitly}.
However, one might wonder whether there are other, faster algorithmic approaches that proceed without
enumerating all moves.
As one of our main contributions, we will show that such faster approaches do not exist for $k=3$ (under the
\textsc{All-Pairs Shortest Paths} conjecture), but do exist for all $k\ge4$.

\subsection{Our contributions}

We investigate whether the long-standing runtimes of~$\Oh(n^2)$ for bitonic tours and~$\Oh(n^k)$ for
finding \kOPT improvements are optimal.
Such optimality investigations usually involve two ingredients: fast algorithms and runtime lower bounds.
While proving unconditional 
lower bounds is far out of reach, in recent years there has been an influx
of techniques for establishing lower bounds on the running time of a given problem, based on a hypothesis about the
best-possible running time for another problem.
Recent results in this direction consider
the problems of computing the \textsc{Longest Common Subsequence}~\cite{AbboudBW15,BringmannK15} of two length-$n$ strings,
the \textsc{Edit Distance} \cite{BackursI15,BringmannK15} from one length-$n$ string to another,
or the \textsc{Discrete Fr\'echet Distance}~\cite{Bringmann14} between two polygonal $n$-vertex curves in the plane.
If one of these problems allows an algorithm with running time $\Oh(n^{2-\eps})$, then this would yield an algorithm to test the satisfiability of an $n$-variable CNF formula~$\phi$
in time $(2-\eps)^{n} \cdot |\phi|^{\Oh(1)}$.
As decades of research have not led to algorithms with such a running time for \cnfsat, this gives evidence that
the classic~$\Oh(n^2)$-time algorithms for these problems are optimal up to~$n^{o(1)}$ factors.

\paragraph{Pyramidal tours in the plane.}
Consider a symmetric \TSP instance that is defined by an edge-weighted complete graph.
For a linear ordering $1,\ldots,n$ of the vertices in the graph, a \emph{pyramidal} tour has the form
$(1,i_1,\ldots,i_r,n,j_1,\ldots,j_{n-r-2})$, where $i_1 < i_2 < \ldots < i_r$ and $j_1 > j_2 > \ldots > j_{n-r-2}$.
A \emph{bitonic} tour for a Euclidean \TSP instance is pyramidal with respect to the left-to-right order on the
points in the plane.
Bitonic and pyramidal tours play an important role in the combinatorial optimization literature on the \TSP;
see \cite{BakiK99,BurkardDDVW98,GilmoreLS85}.
They form an exponentially large set of tours over which we can optimize efficiently, and they lead to well-solvable
special cases of the \TSP.
Combined with a procedure for generating suitable permutations of the vertices, heuristic solutions to \TSP can be
obtained by computing optimal pyramidal tours with respect to the generated orders \cite{CarlierV90}.

We will show that the classic $\Oh(n^2)$ dynamic program for finding bitonic tours in the Euclidean plane is far from
optimal: by an appropriate use of dynamic geometric data structures, the running time can be reduced to~$\Oh(n\log^2n)$.
To the best of our knowledge, this presents the first improvement in finding bitonic tours since the problem was
popularized in \emph{Introduction to Algorithms}~\cite{CormenLRS01} in 1991.
In fact, we prove the stronger result that an optimal \emph{pyramidal} tour among~$n$ points in the plane can be
computed in~$\Oh(n \log^2 n)$ time with respect to any given linear order on the points.
Our techniques extend to the related \textsc{Bottleneck Pyramidal tsp} problem in the plane, where the goal is
to find a pyramidal tour among the cities that minimizes the length of the longest edge.
We prove that the underlying decision problem (given a linearly ordered set of points and a bottleneck value~$B$, is
there a pyramidal tour of the points whose longest edge has length at most~$B$?) can be solved in~$\Oh(n \log n)$ time,
while the underlying optimization version (given a linearly ordered set of points, compute a bitonic tour that
minimizes the length of the longest edge) can be solved in~$\Oh(n \log^3 n)$ time.
For the decision version of the bottleneck problem, we prove a matching~$\Omega(n \log n)$ time lower bound in the
algebraic computation tree model by a reduction from \textsc{Set Disjointness} with integer inputs~\cite{y-lbact-91}; this reduction even applies to the
bitonic setting where the points are ordered from left to right.

\paragraph{$k$-OPT in the graph setting.}
The complexity of \kOPT has been analyzed using the framework of parameterized complexity theory. Marx~\cite{Marx08a} proved that deciding whether there is a \kmove that improves a given tour is W[1]-hard parameterized by~$k$, giving evidence that there is no algorithm with runtime~$f(k) \cdot n^{\Oh(1)}$. Guo~\etal~\cite{GuoHNS13} refined this result and proved that, under the \emph{Exponential Time Hypothesis}~\cite{ImpagliazzoPZ01}, there is no algorithm that determines whether a tour in a weighted complete graph can be improved by a \kmove in time~$f(k) \cdot n^{o(k / \log k)}$ for any function~$f$. This lower bound shows that the exponent of~$n$ in the runtime of any \kOPT algorithm must grow almost linearly with~$k$.
The next question that we settle in this paper is: can one do better than~$\Oh(n^k)$ for finding a \kOPT improvement?
The answer turns out to depend on the value of~$k$.
For \TwoOPT, an easy adversarial argument shows that any deterministic algorithm must inspect all the edge weights.
This gives a trivial lower bound of~$\Omega(n^2)$, matching the upper bound.
For larger values of~$k$, the question becomes more interesting.

The \ThreeOPTImprovement problem asks whether the weight of a given tour can be reduced by some \Threemove.
We show that it is unlikely that \ThreeOPTImprovement with weights in the range~$[-M,\ldots,M]$ allows an algorithm with a \emph{truly subcubic} runtime
of~$\Oh(n^{3-\eps} \polylog (M))$ for~$\eps>0$. We prove that the \NegativeTriangle problem (given an edge-weighted graph,
is there a triangle of negative weight?) reduces to \ThreeOPTImprovement by a reduction that takes~$\Oh(n^2)$ time and increases the size of the graph by only a constant factor.
As \textsc{Negative Edge-Weighted Triangle} is equivalent to \textsc{All-Pairs Shortest Paths} in weighted digraphs (\APSP)
with respect to having truly subcubic algorithms~\cite{WilliamsW10}, a truly subcubic algorithm for \ThreeOPTImprovement would contradict
the \APSP conjecture~\cite{AbboudGW15,AbboudWY15} which states that \APSP cannot be solved in truly subcubic time.
We also give a reduction in the other direction: finding a \ThreeOPT improvement reduces to finding
a negative edge-weighted triangle.
Consequently, \ThreeOPTImprovement is \emph{equivalent} to \textsc{Negative Edge-Weighted Triangle} and \APSP
with respect to truly subcubic runtimes.
This adds yet another classic problem to the growing list of such equivalent problems~\cite{AbboudGW15,WilliamsW10}.

As a final result in this direction, we design an algorithm that finds the best \kOPT improvement in weighted
$n$-vertex complete graphs in $\Oh(n^{\lfloor2k/3\rfloor+1})$ time for each fixed value of $k$.
For $k=2$ and $k=3$, this expression simply boils down to the straightforward time complexities of $\Oh(n^2)$
and $\Oh(n^3)$ for \TwoOPT and \ThreeOPT respectively.
For $k\ge4$, however, our result yields a substantial improvement over the trivial $\Oh(n^k)$ time bound.
For example, {\textsc{$4$-opt}\xspace} can be solved in $\Theta(n^3)$ time, matching the
best-known algorithm for \ThreeOPT.
The algorithm mixes enumeration of partial solutions with a simple dynamic program.

\paragraph{Faster $2$-OPT in the repeated setting and in the planar setting.}
For the \TwoOPT problem in graphs, the runtime for finding a single tour improvement cannot be improved below the trivial~$\Theta(n^2)$. However, in the context of local search we are often interested
in \emph{repeatedly} finding tour improvements. It is therefore natural to consider whether
speedups can be obtained when repeatedly finding improving tours on
the same \TSP instance. We prove that this is indeed the case: after~$\Oh(n^2)$ preprocessing time,
one can repeatedly find the best \TwoOPT improvement in $\Oh(n \log{n})$ time per iteration.

The quadratic lower bound for \TwoOPT applies only in the graph setting. This raises the
question: can we solve \TwoOPT faster for points in the plane? We show
the answer is yes, by giving an algorithm for \TwoOptDET with runtime
$\Oh(n^{8/5+\eps})$ for all~$\eps > 0$. Similarly, \ThreeOptDET can be solved in expected time $\Oh(n^{80/31+\eps})$.

\section{Faster pyramidal TSP}
\label{se:pyramidal}
In this section we show that the pyramidal
\TSP and the bottleneck pyramidal \TSP problem in the plane can be solved in subquadratic time. For simplicity we only show how to compute the value of an optimal solution; computing the actual tour can be done in a standard manner.
\medskip

Let $P$ be the ordered input set of $n$ points with distinct $x$-coordinates in the plane.
Our algorithm will consider the points in $P$ in order, and maintain a collection
of partial solutions that are locally optimal. To make this precise,
define $P_i := \{p_1,\ldots,p_i\}$ to be the first $i$ points in~$P$.
A \emph{partial solution} for $P_i$, for some $1\leq i\leq n$,
is a pair $P',P''$ of monotone paths (with respect to the order on~$P$) that together visit all the points in $P_i$
and that only share~$p_1$. We call a partial
solution for $P_i$ an \emph{$(i,j)$-partial tour}, for some $1\leq j<i$,
if one of the paths ends at~$p_i$---this is necessarily the case
in a partial solution for $P_i$---and the other path ends at~$p_j$.

Our starting point is the standard dynamic-programming solution. It
uses a 2-dimensional table\footnote{Some of our results can also be obtained from an alternative DP with~$n$ states. As we need the 2-dimensional approach for Theorem~\ref{th:bottleneck}, we present all our results in this setting.\label{fn:onedimensionaldp}} $A[1..n,1..n]$, where $A[i,j]$, for $1\leq j<i\leq n$,
is defined as the minimum length of an $(i,j)$-partial tour;
for $i\leq j\leq n$ the entries $A[i,j]$ are undefined.
We can compute the entries in the table row by row, using the recursive formula
\begin{equation} \label{eq:pyramidal:recurrence}
A[i+1,j] \ = \ \left\{ \begin{array}{ll}
                    A[i,j] + |p_i p_{i+1}| & \mbox{if $1\leq j < i$} \\[2mm]
                    \min_{1\leq k<i} \left( A[i,k] + |p_k p_{i+1}| \right) & \mbox{if $j=i$}
                    \end{array}
             \right.
\end{equation}
where $A[2,1]=|p_1 p_2|$. Let us briefly verify this recurrence. For $(i+1,j)$-partial tours with~$j < i$, the path~$P'$ that visits~$p_{i+1}$ must also visit~$p_i$: the other path~$P''$ ends at index~$j < i$ and the monotonicity requirement ensures~$P''$ cannot visit~$i$ and go back to~$j$. So for~$j < i$ any $(i+1,j)$-partial tour consists of an $(i,j)$-partial tour together with the segment~$p_i p_{i+1}$. For $(i+1,i)$-partial tours, the predecessor of~$p_{i+1}$ cannot be~$p_i$, since a path ends at~$p_i$. Hence an $(i+1,i)$-partial tour consists of an $(i,k)$-partial tour for some~$1 \leq k < i$ together with the segment~$p_k p_{i+1}$. The cheapest combination yields the best partial tour.

After computing the last row of~$A$, the minimum length of a pyramidal tour can be
found by computing $\min_{1\leq k<n} \left( A[n,k] + |p_kp_{n}| \right)$. There are~$\Oh(n^2)$ entries in~$A$ of the first type that each take constant time to evaluate. There are~$\Oh(n)$ entries of the second type that need time~$\Theta(n)$. Hence the dynamic program can be evaluated in~$\Oh(n^2)$ time. 


\medskip

Our subquadratic algorithm is based on the following two observations.
First, any two subsequent rows $A[i,1..n]$ and $A[i+1,1..n]$ are quite similar:
the entries $A[i+1,j]$, for $j<i$, can all be obtained from $A[i,j]$ by adding
the same value, namely $|p_i p_{i+1}|$. Second, the computation of
$A[i+1,i]$ can be sped up using appropriate geometric data structures.
Thus our algorithm will maintain a data structure that implicitly represents the
current row and allows for fast queries and so-called bulk updates (see below).

Recall that $P_i := \{p_1,\ldots,p_i\}$. The point that defines
$\min_{1\leq k<i} \left( A[i,k] + |p_k p_{i+1}| \right)$ is the point $p_k\in P_{i-1}$
closest to the query point~$q:= p_{i+1}$ if we use the additively weighted distance function
\begin{equation}\label{eq:additive-weight}
\dist(p_k,q) := w_k + |p_k q|,
\end{equation}
where $w_k := A[i,k]$ is the weight of~$p_k$. Thus we need
a data structure for storing a weighted point set that supports the following operations:
\begin{itemize}
\item perform a \emph{nearest-neighbor query} with a query point~$q$, which reports the
      point~$p_k$ closest to $q$ according to the additively weighted distance function,
\item perform a \emph{bulk update} of the weights, which adds a given
      value~$\Delta$ to the weights of all the points currently stored in the data structure;
\item \emph{insert} a new point with a given weight into the data structure.
\end{itemize}

Answering nearest-neighbor queries for the weighted point set $P$ can be done by performing point location
in the \emph{additively weighted Voronoi diagram}~\cite{Fortune87} of $P$ augmented by a point location data structure~\cite{Snoeyink04PL}. This (static) data structure has size $\Oh(n)$, can be computed in $\Oh(n\log n)$ time, and allows for
$\Oh(\log n)$-time queries. To allow for insertions we use the logarithmic method~\cite{bs-dspsd-80}. The logarithmic method makes a data structure semi-dynamic by storing $\Oh(\log n)$ static data structures of increasing size (resulting in an additional $\log$-factor in the query time). The main observation is that we can handle bulk updates by storing a correction term for the weights with each of the static additively weighted Voronoi diagrams. The additively-weighted nearest neighbor structure does not change when adding the same constant to each point weight, which means we do not have to update the Voronoi diagrams when performing bulk updates. This leads to an implementation that supports each operation in~$\Oh(\log^2 n)$ amortized time. The details are given in Appendix~\ref{se:logarithmic-method}.
Using the data structure we obtain the following theorem.
\begin{theorem}\label{th:regular}
Let $P$ be an ordered set of $n$ points in the plane. Then we can compute a
minimum-length pyramidal tour for $P$ in $\Oh(n\log^2 n)$ time and using $\Oh(n)$ storage.
\end{theorem}
\begin{proof}
We aim to speed up the classic dynamic-programming algorithm using the data structure described above. Instead of computing the entire dynamic programming table~$A$ explicitly, we maintain an implicit representation of one row of the table and compute the rows one by one. The $i$-th row of~$A$ has~$i-1$ well-defined entries. We define an implicit representation of row~$i$ to be an instance of the data structure storing the weighted point set~$P_{i-1} = \{p_1, \ldots, p_{i-1}\}$ such that~$w(p_j) = A[i,j]$. The first nontrivial row in~$A$ is the second row,~$A[2,1..n]$. An implicit representation for that row consists of the point~$p_1$ of weight~$A[2,1] = |p_1 p_2|$. 

If we have an implicit representation of row~$i$, we can efficiently obtain an implicit representation of row~$i+1$, as we describe next. By our choice of implicit representation, the value of~$A[i+1,i]$ according to~(\ref{eq:pyramidal:recurrence}) is exactly the distance from~$p_{i+1}$ to its closest neighbor in the data structure under the additively weighted distance function. Hence, the value of~$k$ that minimizes the lower expression in~(\ref{eq:pyramidal:recurrence}) can be found by a nearest neighbor query with~$p_{i+1}$. We can therefore transform a representation of row~$i$ into a representation for row~$i+1$ as follows:

\begin{enumerate}
	\item Query with point~$p_{i+1}$ to find the value~$A[i+1,i]$ and remember this value.
	\item Perform a bulk update to increase the weight of the points~$p_1, \ldots, p_{i-1}$ that are already in the structure by~$\Delta := |p_i p_{i+1}|$. Recall that for cells~$j$ with~$1 \leq j < i$ their value in row~$i+1$ is obtained from their value in row~$i$ by adding~$|p_i p_{i+1}|$.
	\item Insert point~$p_i$ of weight~$A[i+1, i]$ into the structure.\footnote{We could also insert~$p_i$ with weight~$A[i+1,i]-\Delta$. This way we would not have to subtract~$\Delta$ from the weights of $p_1,\ldots,p_{i-1}$ in Step~2, and the bulk updates are not needed. As they are trivial in our data structure, we prefer the version that keeps the correspondence between weights and $A[i,j]$ values.\label{fn:bulkupdates}}
\end{enumerate}
It is easy to verify that this yields an implicit representation of row~$i+1$. Since a representation of the first nontrivial row can be found in constant time, and each successive row can be computed from the previous using three data structure operations that take~$\Oh(\log ^2 n)$ amortized time each, it follows that an implicit representation of the final row can be computed in~$\Oh(n \log^2 n)$ time. The minimum cost of a pyramidal tour is $\min_{1\leq k<n} \left( A[n,k] + |p_kp_{n}| \right)$, which can be found by querying the representation of the final row with point~$p_n$.
%
\end{proof}



\paragraph{Bottleneck pyramidal TSP.}
Using a similar global approach but different supporting data structures we can also
solve the bottleneck version of the problem---here the goal
is to minimize the length of the longest edge in the tour---in
subquadratic time. 
%
%
For the decision version of the problem we need the following result.
%
\begin{theorem}\label{th:pl-in-disk-union}
We can maintain a collection $\D$ of $n$ congruent disks in a data structure such that
we can decide in $\Oh(\log n)$ time if a query point $q$ lies in $\union(\D)$.
The data structure uses $\Oh(n)$ storage and a new disk can be inserted into $\D$ in $\Oh(\log n)$
amortized time.
\end{theorem}
This result is obtained as follows; 
see Appendix~\ref{subse:bottleneck-decision-alg} for details.
Assume the disks have radius~$\sqrt{2}$ and consider the integer grid.
Let $\D(C)\subseteq \D$ be the set of disks whose centers lie inside a grid cell~$C$.
To decide if $q\in \union(\D)$ we need to test if $q\in \union(\D(C))$ for $O(1)$
grid cells~$C$ that are sufficiently close to~$q$. Now consider a cell~$C$
with $\D(C)\neq \emptyset$. Obviously $C$ itself is completely covered by~$\union(\D(C))$.
Let $\ell_{\mytop}(C)$ be the line containing the top edge of~$C$. 
Then the part of $\union(\D(C))$ above $\ell_{\mytop}(C)$---the other parts
are handled similarly---is $x$-monotone.
Moreover, we can show that each disk $D_i\in\D(C)$ contributes at most
one arc to the boundary of $\union(\D(C))$ above $\ell_{\mytop}(C)$, and the
left-to-right order of the contributed arcs is consistent with the left-to-right
order of the corresponding disk centers. Using this fact, we can do point locations and
insertions in $O(\log n)$ time.

Combining the global technique of the previous section with Theorem~\ref{th:pl-in-disk-union}
we obtain the following theorem.

\begin{theorem}\label{th:bottleneck-decision}
Let $P$ be an ordered set of $n$ points in the plane, and let $B>0$ be a given parameter.
Then we can decide in $\Oh(n\log n)$ time and using $\Oh(n)$ storage
if $P$ admits a pyramidal tour whose longest edge has length at most~$B$. This problem requires $\Omega(n\log n)$ time in the algebraic computation tree model of computation.
\end{theorem}

The algorithm for the decision version does not easily extend to solve the minimization version of the problem.
We therefore design a specialized data structure---a tree storing unions of disks and (regular) Voronoi diagrams---that
allows us to obtain the following result. (See Appendix~\ref{subse:bottleneck}.)

\begin{theorem}\label{th:bottleneck}
Let $P$ be an ordered set of $n$ points in the plane.
Then we can compute a pyramidal tour whose bottleneck edge has minimum
length in $O (n\log^3 n)$ time and using $\Oh(n\log n)$ storage.
\end{theorem}


\section{The $k$-OPT problem in general graphs}
In this section we change the perspective from Euclidean problems to the \TSP in general graphs. A \emph{tour} of an undirected graph~$G$ is a Hamiltonian cycle in the graph. Depending on the context, we may treat a tour as a permutation of the vertex set or as a set of edges. We consider undirected, weighted complete graphs to model symmetric TSP inputs. The weight of a tour is simply the sum of the weights of its edges. Recall that a \emph{\kmove} of a tour~$T$ is an operation that replaces a set of $k$ edges in~$T$ by another set
of $k$ edges from $G$ in such a way that the result is a valid tour. In degenerate cases, such an operation may delete and reinsert the same edge.
The associated decision problem is defined as follows.

\defproblem{\kOPTDetection}
{A complete undirected graph~$G$ along with a (symmetric) distance function~$d \colon E(G) \to \mathbb{N}$, an integer~$k$, and a tour~$T \subseteq E(G)$.}
{Is there a \kmove that strictly improves the cost of~$T$?}

The optimization problem \kOPTOptimization is to compute, given a tour in a graph, a \kmove that gives the largest cost improvement, or report that no improving \kmove exists.

\subsection{On truly subcubic algorithms for 3-OPT}

We say that an algorithm for $n$-vertex graphs with integer edge weights in the range~$[-M,\ldots,M]$ runs in \emph{truly subcubic time} if its runtime is bounded by~$\Oh(n^{3-\eps} \polylog(M))$ for some constant~$\eps > 0$. Vassilevska-Williams and Williams~\cite{WilliamsW10} introduced a framework for relating the truly subcubic solvability of several classic problems to each other. We use it to show that the existence of a truly subcubic algorithm for \ThreeOPT is unlikely. Their framework uses a notion of subcubic reducibility based on Turing reducibility~\cite[\S IV]{WilliamsW10} that solves one instance of problem~$A$ by repeatedly solving inputs of problem~$B$. For our applications, simple reductions suffice that transform one input of problem~$A$ into one input of problem~$B$ of roughly the same size, in~$\Oh(n^2)$ time.\footnote{We assume that simple arithmetic on weights can be done in constant time. The $\polylog(M)$ factors used in the framework originate from repeated executions to perform binary search on weight values.} Such reductions preserve the existence of truly subcubic algorithms, so we take this simpler viewpoint. The following problem is the starting point for our reductions.

\defproblem
{\NegativeTriangle}
{An undirected, complete graph~$G$ and a weight function $w \colon E(G) \to\mathbb{Z}$.}
{Does $G$ contain a triangle whose total edge-weight is negative?}

Vassilevska-Williams and Williams~\cite[Thm. 1.1]{WilliamsW10} proved that \NegativeTriangle has a truly subcubic algorithm if and only if the \textsc{All-Pairs Shortest Paths} problem on digraphs with non-negative integral edge weights has a truly subcubic algorithm.

\begin{lemma} \label{lemma:negtriangle:to:threeopt}
\NegativeTriangle can be reduced to \ThreeOPTDetection in time~$\Oh(n^2)$ while increasing the size of the graph and the largest weight by a constant factor.
\end{lemma}
\begin{proof}
Consider an instance $(G,w)$ of \NegativeTriangle, and let $v_1,\ldots,v_n$ be an enumeration
of the vertices of~$G$.
Let~$M$ be the largest absolute value of an edge weight.
We introduce an instance of \ThreeOPTImprovement that consists of $2n$ vertices $a_1,\ldots,a_n$ and
$b_1,\ldots,b_n$, where the starting tour $T$ uses the ordering $a_1,b_1,a_2,b_2,\ldots,a_n,b_n$.
The (symmetric) distances $d(\cdot,\cdot)$ between these vertices are defined as follows:
\begin{itemize}
\item $d(a_i,b_i)=0$ for $1\le i\le n$;
\item $d(b_n,a_1)=-3M$, and $d(b_i,a_{i+1})=-3M$ for $1\le i\le n-1$;
\item $d(a_i,b_j)=w(\{v_i,v_j\})$ for $1\le i < j\le n$;
\item $d(b_i,a_j)=w(\{v_i,v_j\})$ for $1\le i < j - 1 \le n - 1$;
\item $d(a_i,a_j)=d(b_i,b_j)=3M$ for $1\le i\ne j\le n$.
\end{itemize}
(For convenience, we allow distances to be negative in this construction.
One easily moves to non-negative distances by adding the constant $4M$ to all distances.)

\begin{claim}
The constructed instance of \ThreeOPTImprovement allows an improving \ThreeOPT move,
if and only if the graph $G$ contains a triangle of negative edge-weight.
\end{claim}
\begin{claimproof}
($\Leftarrow$) Assume that the vertices $v_i,v_j,v_k$ span a triangle of negative edge-weight in $G$ for~$i<j<k$.
We remove the three edges $\{a_i,b_i\}$, $\{a_j,b_j\}$, and $\{a_k,b_k\}$ from tour $T$,
and we reconnect the resulting pieces by the three edges $\{a_i,b_j\}$, $\{a_j,b_k\}$, and $\{a_k,b_i\}$.
The three removed edges have total length~$0$, while the three inserted edges have negative total length.

($\Rightarrow$) Now assume that there exists an improving \Threemove for tour $T$.
This improving move cannot remove any edge $\{b_i,a_{i+1}\}$ or $\{b_n,a_1\}$, as these edges have length $-3M$, the tour~$T$ contains no edges of positive length to potentially remove, and each edge that enters the tour has length at least~$-M$.
Consequently, the three removed edges will be $\{a_i,b_i\}$, $\{a_j,b_j\}$, and $\{a_k,b_k\}$ for some $i<j<k$.
As these three edges have total length~$0$, the total length of the three inserted edges must be strictly negative.
The edges $\{a_x,a_y\}$ and $\{b_x,b_y\}$ all have length $3M$, while the edges $\{a_x,b_y\}$ all have length
between $-M$ and $M$.
This implies that every inserted edge is either of the type $\{a_x,b_y\}$, or coincides with one of the removed edges.
Suppose for the sake of contradiction that one of the inserted edges coincides with a removed edge $\{a_k,b_k\}$,
so that we are actually dealing with a \Twomove.
Then the two inserted edges in the \Twomove must be $\{a_i,a_j\}$ and $\{b_i,b_j\}$, so that the new tour
is by $6M$ \emph{longer} than the old tour $T$.
This contradiction leaves only two possibilities for the three inserted edges:
either $\{a_i,b_j\}$, $\{a_j,b_k\}$, $\{a_k,b_i\}$, or $\{a_i,b_k\}$, $\{a_k,b_j\}$, $\{a_j,b_i\}$ (of which the latter is actually not a valid 3-move). Since the total length of the three inserted edges is strictly negative, the three vertices $v_i,v_j,v_k$
form a triangle of strictly negative weight in $G$.
\end{claimproof}

The claim shows the correctness of the reduction. It is easy to perform in~$\Oh(n^2)$ time.
\end{proof}

There is an analogous reduction in the other direction, which is given as Lemma~\ref{lemma:threeopt:to:negtriangle} in Appendix~\ref{se:appendix-reduction}. Together, these lemmata show the equivalence of finding negative-weight triangles and detecting improving \ThreeOPT moves. From our reductions and the results of Vassilevska-Williams and Williams~\cite[Thm. 1.1]{WilliamsW10}, we obtain the following theorem.

\begin{theorem}
There is a truly subcubic algorithm for \ThreeOPTDetection if and only if there is such an algorithm for \textsc{All-Pairs Shortest Paths} on weighted digraphs.
\end{theorem}

\subsection{A fast $k$-OPT algorithm}
In this section we will prove that the \kOPTOptimization problem can be solved significantly faster than $\Theta(n^k)$ when $k\geq 4$. To this end, we first analyze the structure of \kOPT moves.
Consider a \kmove for a given tour $T\subseteq E(G)$, and let $e_1,\ldots,e_k$ be the removed edges
with $e_i=\{v_{2i-1},v_{2i}\}$.
We assume throughout that these vertices (and edges) are indexed in such a way that $T$ traverses the
vertices $v_i$ in order of increasing index.
We assume furthermore that the vertices $v_1,\ldots,v_{2n}$ are pairwise distinct. By a simple reduction that subdivides the edges on the tour (Appendix~\ref{se:appendix-kopt}), it is sufficient to find an algorithm that detects an improving \kmove in which the removed edges do not share any endpoints. (The arguments presented here also
go through without this assumption, but the notation becomes more complicated in the equality case.)
The $k$ new edges that are inserted into $T$ are denoted $f_1,\ldots,f_k$.
The \emph{signature} of this \kmove is a permutation $\pi$ of~$\{1,\ldots,2k\}$,
such that vertex $v_j$ and vertex
$v_{\pi(j)}$ form the endpoints of one of the edges $f_1,\ldots,f_k$;
see Fig.~\ref{fi:4-changes}.
\begin{figure}
\begin{center}
\includegraphics{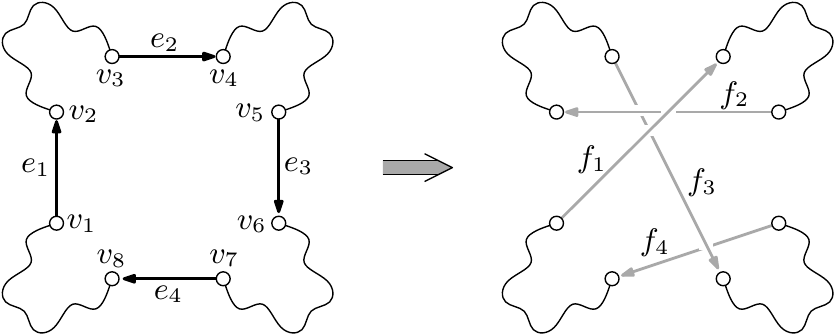}
\end{center}
\caption{A 4-change with signature 4,5,7,1,2,8,3,6.
         Edges $e_1$ and $e_4$ are non-interfering. As we work on symmetric TSP, the graph and distance function are undirected; the arc directions merely indicate the traversal direction with respect to an arbitrary orientation of the tour.}
\label{fi:4-changes}
\end{figure}
Note that the removed edges $e_1,\ldots,e_k$ together with the signature $\pi$ fully determine
the \kmove (and in particular determine the inserted edges $f_1,\ldots,f_k$).

Note furthermore that not every permutation $\pi$ yields a feasible signature that corresponds
to some \kmove:
First, in a feasible signature $\pi(i)=j$ always implies $\pi(j)=i$, and we will always have $\pi(i)\ne i$.
Secondly, in a feasible signature the edge set that results from $T$ by removing $e_1,\ldots,e_k$ and by
inserting $f_1,\ldots,f_k$ must form a single Hamiltonian cycle---it must never form a collection of two or
more cycles.
It is easy to check whether a given permutation $\pi$ constitutes a feasible
signature, and to enumerate all feasible signatures.

We say that two of the removed edges $e_i$ and $e_j$ \emph{interfere} with each other in a \kmove, if
there exists an inserted edge $f$ that connects one of the endpoints of $e_i$ to an endpoint of $e_j$. The following lemma states that in any \kmove, there is a set of~$\lceil k/3\rceil$ pairwise non-interfering edges. This is essentially due to the fact that every $k$-vertex $2$-regular graph (collection of cycles) contains an independent set of size at least~$\lceil k/3 \rceil$; we prove it here in the \kOPT terminology.
\begin{lemma}
For any signature $\pi$, we can find a subset $E_{\pi}\subseteq\{e_1,\ldots,e_k\}$ of at least
$\lceil k/3\rceil$ removed edges that are pairwise non-interfering.
\end{lemma}
\begin{proof}
The $2k$ edges $e_1,\ldots,e_k$ and $f_1,\ldots,f_k$ induce a set of cycles on the vertices $v_1,\ldots,v_{2k}$.
If such a cycle contains an even number of removed edges, say $2\ell$,
we put every other removed edge along this
cycle into $E_{\pi}$; this yields $\ell$ out of $2\ell$ edges for $E_{\pi}$.
If the cycle contains only a single removed edge, we put this single edge into $E_{\pi}$; this yields one out
of one edge for $E_{\pi}$.
If the cycle contains an odd number of removed edges, say $2\ell+1\ge3$,
we ignore the first removed edge and then
put every other removed edge along the cycle into $E_{\pi}$; this yields $\ell$ out of $2\ell+1$ edges for $E_{\pi}$.
The weakest contribution to $E_{\pi}$ comes from cycles with three removed edges, which yield only one out of
three edges for $E_{\pi}$. The claimed bound $\lceil k/3\rceil$ follows.
\end{proof}
\begin{theorem}
For every fixed~$k \geq 3$, the \kOPTOptimization problem on an $n$-vertex graph can be solved in~$\Oh(n^{\lfloor2k/3\rfloor+1})$ time.
\end{theorem}
\begin{proof}
For computing the best \kmove for tour $T$, it is sufficient to compute for
every feasible signature $\pi$---for fixed $k$ there are only $\Oh(1)$
such signatures---the best \kmove for tour~$T$ with that particular signature.
This is done as follows.
We first determine a set $E_{\pi}$ of pairwise non-interfering edges according to the above lemma.
Then we enumerate and handle all possible cases for the locations of the $\lfloor2k/3\rfloor$ removed edges
not in $E_{\pi}$ along $T$.
This yields $\Oh(n^{\lfloor2k/3\rfloor})$ cases to handle, and every such case will be handled in $\Oh(n)$
time; note that this yields the claimed complexity.
In handling a case, the positions of the removed edges not in $E_{\pi}$ are frozen, while the edges in $E_{\pi}$
have to be embedded into $T$.
The cost of a \kmove with signature $\pi$ decomposes into two parts:
\begin{itemize}
\item
The first part consists of the total weight of all frozen edges (which is subtracted) and the total weight
of inserted edges between frozen edges (which is added).
\item
The second part consists of the individual contributions of the edges in $E_{\pi}$.
For an edge $e\in E_{\pi}$ and an edge $e'\in T$, the cost of embedding $e$ into $e'$ equals the weight of the
two inserted edges adjacent to $e$ minus the weight of $e'$.
As the edges in $E_{\pi}$ are pairwise non-interfering, their individual cost contributions do not interact
with each other.
\end{itemize}
As the cost of the first part is fixed in every considered case, our goal is to minimize the total cost
of the second part.
The frozen edges subdivide the tour $T$ into a number of tour pieces, and we have to find the cheapest way
of embedding the corresponding edges from $E_{\pi}$ into such a tour piece.
The following paragraph sketches a straightforward dynamic program for finding the optimal embedding for
each tour piece in time proportional to the length of the piece.
As the length of all tour pieces combined is $\Oh(n)$, every case is indeed handled in time $\Oh(n)$.

We are essentially dealing with the following optimization problem.
There are $r$ locations $L_1,\ldots,L_r$ (the edges along tour $T$ between two consecutive frozen edges) and
$s$ objects $O_1,\ldots,O_s$ (the edges in $E_{\pi}$ that should be embedded between the two considered frozen edges).
The objects are to be embedded into the locations, so that the location of object $O_i$ always precedes the
location of object $O_{i+1}$.
The cost of embedding object $O_i$ into location $L_j$ is denoted $c(i,j)$.
For~$1 \leq x \leq s$ and~$1 \leq y \leq r$, let~$V(x,y)$ denote the smallest possible cost incurred by embedding the first~$x$ objects~$O_1, \ldots, O_x$ into the first~$y$ locations~$L_1, \ldots, L_y$.
As $V(x,y)$ equals the minimum of $V(x,y-1)$ and $V(x-1,y-1)+c(x,y)$, all these values $V(x,y)$ can easily
be computed in $\Oh(rs)$ time.
In our situation, $r$ is the length of the considered tour piece and $s\le k$ is a constant that does not depend
on the input; hence the complexity is indeed proportional to the length of the considered tour piece.
\end{proof} 
\section{Faster 2-OPT}
In this section we show that it is possible to beat the quadratic barrier for \TwoOPT
in two important settings, namely when we want to
apply \Twomoves repeatedly, and in the Euclidean setting in the plane.

\paragraph{Repeated 2-OPT.}
In the repeated \TwoOPT problem, we apply \TwoOPT repeatedly (e.g. until no further improvements
are possible). One can considerably speed up the \TwoOPT computations
at each of the iterations, except the first one. The following theorem gives our improvement for the
\TwoOptOPT problem, where the goal is to find the best \Twomove
(rather than any \Twomove that improves the tour).

\begin{theorem} \label{theorem:repeated:twoopt}
After $\Oh(n^2)$ preprocessing and using $\Oh(n^2)$ storage we can repeatedly solve
the \TwoOptOPT problem in $\Oh(n\log n)$ time per iteration.
\end{theorem}

The speedup claimed in the theorem relies on a tour representation that supports efficient \Twomoves. To apply a \Twomove that removes two edges~$e$ and~$e'$ and replaces them by the appropriate diagonal connections, one effectively has to reverse the part of the tour between~$e$ and~$e'$, or the part between~$e'$ and~$e$. It can therefore take~$\Omega(n)$ time to apply a \Twomove to a tour represented as a sequence of vertices in an array. Chrobak~\etal~\cite{ChrobakSK99} give a speedup by storing the cities on the tour in an ordered balanced binary search tree. Each node in the tree stores a bit indicating whether the tour order is given by an in-order traversal of the subtree rooted there, or by the \emph{reverse} of the in-order traversal. This allows a \Twomove to be applied in~$\Oh(\log n)$ time by manipulating reversal bits.

Our approach for repeated \TwoOptOPT is based on a similar data structure that represents tours in balanced search trees. However, instead of having only one tree that stores the current tour, we have~$n$ trees; one for each edge~$e_1, \ldots, e_n$ in the current tour. A query in the tree~$\tree(e_i)$ corresponding to edge~$e_i$ can be used to determine which edge~$e_j$ yields the most profitable \Twomove together with~$e_i$. After initializing these~$n$ trees, which takes~$\Oh(n^2)$ time, an iteration of \TwoOptOPT can be performed as follows. For each~$e_i$ on the current tour, we query in tree~$\tree(e_i)$ to find the best \Twomove that removes~$e_i$ and some unknown edge~$e_j$ in~$\Oh(\log n)$ time. In this way we find the best overall \Twomove which removes, say, edges~$e_i$ and~$e_j$. We can update all trees~$\tree(e_\ell)$ for~$\ell \neq i,j$ by deleting $e_i$ and $e_j$, and inserting the appropriate replacement edges. Using the reversal bits this can be done in $O(\log n)$ time. Trees~$\tree(e_i)$ and~$\tree(e_j)$ are destroyed; we build two new trees from scratch for the two new edges~$e_{i'}$ and~$e_{j'}$ that enter the tour. This gives~$\Oh(n \log n)$ time per iteration.

It is likely that these techniques can be extended to speed up repeated \ThreeOPT as well. As the technical details become substantially more cumbersome, we do not pursue this direction.

\paragraph{The planar case.}
For points in the plane (and under the Euclidean metric) we can speed up \TwoOPT computations by using suitable geometric data structures for semi-algebraic range searching, as shown in Appendix~\ref{subse:appendix-planar-case}.
(Note that we do not consider the repeated version of the problem, but the single-shot version.)
A similar approach can be used to speed up 3-OPT in the Euclidean setting in the plane. This leads to the following theorem.
\begin{theorem} \label{thm:twodet:threedet:planar}
For any fixed $\eps>0$,
\TwoOptDET in the plane can be solved in $\Oh(n^{8/5+\eps})$ time,
and \ThreeOptDET in the plane can be solved in $\Oh(n^{80/31+\eps})$ expected time.
\end{theorem}

\section{Conclusion}

Revisiting the worst-case complexity of \kOPT and pyramidal \TSP led to a number of new results on these classic problems. Some, such as the equivalence between \ThreeOPT and \APSP with respect to having truly subcubic algorithms, rely on very recent work. Other results, such as the near-linear time algorithm for finding bitonic tours, and the \kOPT algorithm that beats the trivial~$\Oh(n^k)$ upper bound, are obtained using classic techniques. In this respect, it is surprising that these results were not found earlier. These examples show that the availability of new lower bound machinery can inspire new algorithms.

Our findings suggest several directions for further research, both theoretical and applied. An interesting open problem regarding \kOPTDetection is whether the problem is fixed-parameter tractable when improving a given tour in an edge-weighted planar graph. This question was also asked by Marx~\cite{Marx08a} and Guo et al.~\cite{GuoHNS13}. Similarly, it is open whether the problem is fixed-parameter tractable when improving a given tour among points in the Euclidean plane. It would be interesting to settle the exact complexity of \kOPT in general weighted graphs. Is~$\Theta(n^{\lfloor \frac{2k}{3} \rfloor + 1})$ the optimal running time for \kOPTDetection? When all weights lie in the range~$[-M, \ldots, M]$, one can detect a negative triangle in an edge-weighted graph in time~$\Oh(M \cdot n^{\omega})$ using fast matrix multiplication~\cite{AlonGM97,RodittyW11,Yuval76}. By our reduction, this gives an algorithm for \ThreeOPTImprovement with weights~$[-M, \ldots, M]$ in time~$\Oh(M \cdot n^{\omega})$. Can similar speedups be obtained for \kOPT for larger~$k$?

Given the great industrial interest in \TSP, establishing the practical applicability of these theoretical results is an important follow-up step. Several of our results rely on data structures that are efficient in theory, but which are currently impractical. These include the additively-weighted Voronoi diagram used for pyramidal tours on points in the plane, and the semi-algebraic range searching data structures used to speed up \TwoOPTDetection. In contrast, the~$\Oh(n^{\floor{2k/3}+1})$ algorithm for finding the best \kmove improvement is self-contained, easy to implement, and may have practical potential. 


%

\paragraph{Acknowledgments.} We are grateful to Hans L.~Bodlaender, Karl Bringmann, and Jesper Nederlof for insightful discussions, an anonymous referee for the observation in Footnote~\ref{fn:onedimensionaldp}, and Christian Knauer for the observation in Footnote~\ref{fn:bulkupdates}. 

\bibliographystyle{abbrvurl}
\bibliography{Paper}

\begin{thebibliography}{10}

\bibitem{AbboudBW15}
A.~Abboud, A.~Backurs, and V.~V. Williams.
\newblock Tight hardness results for {LCS} and other sequence similarity
  measures.
\newblock In {\em Proc. 56th FOCS}, pages 59--78, 2015.
\newblock \href {http://dx.doi.org/10.1109/FOCS.2015.14}
  {\path{doi:10.1109/FOCS.2015.14}}.

\bibitem{AbboudGW15}
A.~Abboud, F.~Grandoni, and V.~Vassilevska-Williams.
\newblock Subcubic equivalences between graph centrality problems, {APSP} and
  diameter.
\newblock In {\em Proc. 26th SODA}, pages 1681--1697, 2015.
\newblock \href {http://dx.doi.org/10.1137/1.9781611973730.112}
  {\path{doi:10.1137/1.9781611973730.112}}.

\bibitem{AbboudWY15}
A.~Abboud, V.~Vassilevska-Williams, and H.~Yu.
\newblock Matching triangles and basing hardness on an extremely popular
  conjecture.
\newblock In {\em Proc. 47th STOC}, pages 41--50, 2015.
\newblock \href {http://dx.doi.org/10.1145/2746539.2746594}
  {\path{doi:10.1145/2746539.2746594}}.

\bibitem{am-rssas-94}
P.~Agarwal and J.~Matousek.
\newblock On range searching with semialgebraic sets.
\newblock {\em Discr. Comput. Geom}, 11:393--418, 1994.

\bibitem{ams-rssasII-13}
P.~Agarwal, J.~Matousek, and M.~Sharir.
\newblock On range searching with semialgebraic sets, {II}.
\newblock {\em SIAM J. Comput.}, 42:2039--2062, 2013.
\newblock \href {http://dx.doi.org/10.1137/120890855}
  {\path{doi:10.1137/120890855}}.

\bibitem{AlonGM97}
N.~Alon, Z.~Galil, and O.~Margalit.
\newblock On the exponent of the all pairs shortest path problem.
\newblock {\em J. Comput. Syst. Sci.}, 54(2):255--262, 1997.
\newblock \href {http://dx.doi.org/10.1006/jcss.1997.1388}
  {\path{doi:10.1006/jcss.1997.1388}}.

\bibitem{BackursI15}
A.~Backurs and P.~Indyk.
\newblock Edit distance cannot be computed in strongly subquadratic time
  (unless {SETH} is false).
\newblock In {\em Proc. 47th STOC}, pages 51--58, 2015.
\newblock \href {http://dx.doi.org/10.1145/2746539.2746612}
  {\path{doi:10.1145/2746539.2746612}}.

\bibitem{BakiK99}
M.~F. Baki and S.~N. Kabadi.
\newblock Pyramidal traveling salesman problem.
\newblock {\em Computers {\&} {OR}}, 26(4):353--369, 1999.
\newblock \href {http://dx.doi.org/10.1016/S0305-0548(98)00067-7}
  {\path{doi:10.1016/S0305-0548(98)00067-7}}.

\bibitem{Bentley90}
J.~L. Bentley.
\newblock Experiments on traveling salesman heuristics.
\newblock In {\em Proc. 1st SODA}, pages 91--99, 1990.

\bibitem{bs-dspsd-80}
J.~Bently and J.~Saxe.
\newblock Decomposable searching problems {I}: Static-to-dynamic
  transformation.
\newblock {\em J. Algorithms}, 1:301--358, 1980.
\newblock \href {http://dx.doi.org/10.1016/0196-6774(80)90015-2}
  {\path{doi:10.1016/0196-6774(80)90015-2}}.

\bibitem{Bringmann14}
K.~Bringmann.
\newblock Why walking the dog takes time: Frechet distance has no strongly
  subquadratic algorithms unless {SETH} fails.
\newblock In {\em Proc. 55th FOCS}, pages 661--670, 2014.
\newblock \href {http://dx.doi.org/10.1109/FOCS.2014.76}
  {\path{doi:10.1109/FOCS.2014.76}}.

\bibitem{BringmannK15}
K.~Bringmann and M.~K{\"{u}}nnemann.
\newblock Quadratic conditional lower bounds for string problems and dynamic
  time warping.
\newblock In V.~Guruswami, editor, {\em Proc. 56th FOCS}, pages 79--97. {IEEE}
  Computer Society, 2015.
\newblock \href {http://dx.doi.org/10.1109/FOCS.2015.15}
  {\path{doi:10.1109/FOCS.2015.15}}.

\bibitem{BurkardDDVW98}
R.~E. Burkard, V.~G. Deineko, R.~van Dal, J.~A.~A. van~der Veen, and G.~J.
  Woeginger.
\newblock Well-solvable special cases of the traveling salesman problem: {A}
  survey.
\newblock {\em {SIAM} Review}, 40(3):496--546, 1998.
\newblock \href {http://dx.doi.org/10.1137/S0036144596297514}
  {\path{doi:10.1137/S0036144596297514}}.

\bibitem{CarlierV90}
J.~Carlier and P.~Villon.
\newblock A new heuristic for the travelling salesman problem.
\newblock {\em RAIRO -- Operations Research}, 24:245--253, 1990.

\bibitem{ChandraKT99}
B.~Chandra, H.~J. Karloff, and C.~A. Tovey.
\newblock New results on the old $k$-{OPT} algorithm for the traveling salesman
  problem.
\newblock {\em {SIAM} J. Comput.}, 28(6):1998--2029, 1999.
\newblock \href {http://dx.doi.org/10.1137/S0097539793251244}
  {\path{doi:10.1137/S0097539793251244}}.

\bibitem{ChrobakSK99}
M.~Chrobak, T.~Szymacha, and A.~Krawczyk.
\newblock A data structure useful for finding hamiltonian cycles.
\newblock {\em Theoretical Computer Science}, 71(3):419--424, 1990.
\newblock \href {http://dx.doi.org/10.1016/0304-3975(90)90053-K}
  {\path{doi:10.1016/0304-3975(90)90053-K}}.

\bibitem{CormenLRS01}
T.~H. Cormen, C.~E. Leiserson, R.~L. Rivest, and C.~Stein.
\newblock {\em Introduction to Algorithms, Third Edition}.
\newblock The MIT Press, 3rd edition, 2009.

\bibitem{Croes58}
G.~A. Croes.
\newblock A method for solving traveling-salesman problems.
\newblock {\em Operations Research}, 6:791--812, 1958.
\newblock \href {http://dx.doi.org/10.1287/opre.6.6.791}
  {\path{doi:10.1287/opre.6.6.791}}.

\bibitem{EnglertRV14}
M.~Englert, H.~R{\"{o}}glin, and B.~V{\"{o}}cking.
\newblock Worst case and probabilistic analysis of the 2-opt algorithm for the
  {TSP}.
\newblock {\em Algorithmica}, 68(1):190--264, 2014.
\newblock \href {http://dx.doi.org/10.1007/s00453-013-9801-4}
  {\path{doi:10.1007/s00453-013-9801-4}}.

\bibitem{Fortune87}
S.~Fortune.
\newblock A sweepline algorithm for {V}oronoi diagrams.
\newblock {\em Algorithmica}, 2:153--174, 1987.
\newblock \href {http://dx.doi.org/10.1007/BF01840357}
  {\path{doi:10.1007/BF01840357}}.

\bibitem{FredmanJMO95}
M.~L. Fredman, D.~S. Johnson, L.~A. McGeoch, and G.~Ostheimer.
\newblock Data structures for traveling salesmen.
\newblock {\em J. Algorithms}, 18(3):432--479, 1995.
\newblock \href {http://dx.doi.org/10.1006/jagm.1995.1018}
  {\path{doi:10.1006/jagm.1995.1018}}.

\bibitem{GilmoreLS85}
P.~Gilmore, E.~Lawler, and D.~Shmoys.
\newblock Well-solved special cases.
\newblock In E.~Lawler, J.~Lenstra, A.~R. Kan, and D.~Shmoys, editors, {\em The
  Traveling Salesman Problem}, pages 87--143. Wiley, New York, 1985.

\bibitem{Glover96}
F.~Glover.
\newblock Finding a best traveling salesman 4-{Opt} move in the same time as a
  best 2-{Opt} move.
\newblock {\em J. Heuristics}, 2(2):169--179, 1996.
\newblock \href {http://dx.doi.org/10.1007/BF00247211}
  {\path{doi:10.1007/BF00247211}}.

\bibitem{GuoHNS13}
J.~Guo, S.~Hartung, R.~Niedermeier, and O.~Such{\'{y}}.
\newblock The parameterized complexity of local search for {TSP}, more refined.
\newblock {\em Algorithmica}, 67(1):89--110, 2013.
\newblock \href {http://dx.doi.org/10.1007/s00453-012-9685-8}
  {\path{doi:10.1007/s00453-012-9685-8}}.

\bibitem{ImpagliazzoPZ01}
R.~Impagliazzo, R.~Paturi, and F.~Zane.
\newblock Which problems have strongly exponential complexity?
\newblock {\em J. Comput. Syst. Sci.}, 63(4):512--530, 2001.
\newblock \href {http://dx.doi.org/10.1006/jcss.2001.1774}
  {\path{doi:10.1006/jcss.2001.1774}}.

\bibitem{JohnsonMc97}
D.~Johnson and L.~McGeoch.
\newblock The traveling salesman problem: A case study in local optimization.
\newblock In E.~Aarts and J.~Lenstra, editors, {\em Local search in
  combinatorial optimization}, pages 215--310. Wiley, Chichester, 1997.

\bibitem{JohnsonM02}
D.~S. Johnson and L.~A. McGeoch.
\newblock Experimental analysis of heuristics for the {STSP}.
\newblock In G.~Gutin and A.~Punnen, editors, {\em The Traveling Salesman
  Problem and its Variations}, pages 369--443. Kluwer Academic Publishers,
  Dordrecht, 2002.

\bibitem{k-atuvd-04}
V.~Koltun.
\newblock Almost tight upper bounds for vertical decompositions in four
  dimensions.
\newblock {\em J. ACM}, 51:699--730, 2004.
\newblock \href {http://dx.doi.org/10.1145/1017460.1017461}
  {\path{doi:10.1145/1017460.1017461}}.

\bibitem{KunnemannM15}
M.~K{\"{u}}nnemann and B.~Manthey.
\newblock Towards understanding the smoothed approximation ratio of the 2-opt
  heuristic.
\newblock In {\em Proc. 42nd ICALP}, pages 859--871, 2015.
\newblock \href {http://dx.doi.org/10.1007/978-3-662-47672-7_70}
  {\path{doi:10.1007/978-3-662-47672-7_70}}.

\bibitem{Lin65}
S.~Lin.
\newblock Computer solutions of the traveling salesman problem.
\newblock {\em Bell System Technical Journal}, 44(10):2245--2269, 1965.
\newblock \href {http://dx.doi.org/10.1002/j.1538-7305.1965.tb04146.x}
  {\path{doi:10.1002/j.1538-7305.1965.tb04146.x}}.

\bibitem{Marx08a}
D.~Marx.
\newblock Searching the $k$-change neighborhood for {TSP} is {W[1]}-hard.
\newblock {\em Oper. Res. Lett.}, 36(1):31--36, 2008.
\newblock \href {http://dx.doi.org/10.1016/j.orl.2007.02.008}
  {\path{doi:10.1016/j.orl.2007.02.008}}.

\bibitem{m-rsehc-93}
J.~Matou\v{s}ek.
\newblock Range searching with efficient hierarchical cuttings.
\newblock {\em Discr. Comput. Geom.}, 10:157--182, 1993.
\newblock \href {http://dx.doi.org/10.1007/BF02573972}
  {\path{doi:10.1007/BF02573972}}.

\bibitem{MavroidisPP07}
I.~Mavroidis, I.~Papaefstathiou, and D.~N. Pnevmatikatos.
\newblock A fast {FPGA}-based 2-opt solver for small-scale euclidean traveling
  salesman problem.
\newblock In {\em {IEEE} Symposium on Field-Programmable Custom Computing
  Machines}, pages 13--22, 2007.
\newblock \href {http://dx.doi.org/10.1109/FCCM.2007.40}
  {\path{doi:10.1109/FCCM.2007.40}}.

\bibitem{ONeilB15}
M.~A. O'Neil and M.~Burtscher.
\newblock Rethinking the parallelization of random-restart hill climbing: a
  case study in optimizing a 2-opt {TSP} solver for {GPU} execution.
\newblock In {\em Proceedings of the 8th Workshop on General Purpose Processing
  using GPUs}, pages 99--108, 2015.
\newblock \href {http://dx.doi.org/10.1145/2716282.2716287}
  {\path{doi:10.1145/2716282.2716287}}.

\bibitem{RodittyW11}
L.~Roditty and V.~Vassilevska-Williams.
\newblock Minimum weight cycles and triangles: Equivalences and algorithms.
\newblock In {\em Proc. 52nd FOCS}, pages 180--189, 2011.
\newblock \href {http://dx.doi.org/10.1109/FOCS.2011.27}
  {\path{doi:10.1109/FOCS.2011.27}}.

\bibitem{Snoeyink04PL}
J.~Snoeyink.
\newblock Point location.
\newblock In J.~E. Goodman and J.~O'Rourke, editors, {\em Handbook of Discrete
  and Computational Geometry (2nd ed.)}. CRC Press, 2004.

\bibitem{WilliamsW10}
V.~Vassilevska-Williams and R.~Williams.
\newblock Subcubic equivalences between path, matrix and triangle problems.
\newblock In {\em Proc. 51th FOCS}, pages 645--654, 2010.
\newblock \href {http://dx.doi.org/10.1109/FOCS.2010.67}
  {\path{doi:10.1109/FOCS.2010.67}}.

\bibitem{wl-arrdd-85}
D.~Willard and G.~Lueker.
\newblock Adding range restriction capability to dynamic data structures.
\newblock {\em J. ACM}, 32:597--617, 1985.
\newblock \href {http://dx.doi.org/10.1145/3828.3839}
  {\path{doi:10.1145/3828.3839}}.

\bibitem{y-lbact-91}
A.~C.-C. Yao.
\newblock Lower bounds for algebraic computation trees with integer inputs.
\newblock {\em SIAM J. Comput.}, 20(4):655--668, 1991.
\newblock \href {http://dx.doi.org/10.1137/0220041}
  {\path{doi:10.1137/0220041}}.

\bibitem{Yuval76}
G.~Yuval.
\newblock An algorithm for finding all shortest paths using $n^{2.81}$
  infinite-precision multiplications.
\newblock {\em Inf. Process. Lett.}, 4(6):155--156, 1976.
\newblock \href {http://dx.doi.org/10.1016/0020-0190(76)90085-5}
  {\path{doi:10.1016/0020-0190(76)90085-5}}.

\end{thebibliography}

\newpage

\appendix
\section{Data structure for faster pyramidal TSP}
\label{se:logarithmic-method}
In this section we describe the data structure used in Theorem~\ref{th:regular}.
With a slight abuse of notation, we will denote the
set of points stored in the data structure by~$P$ and let $n$ denote the number
of points in the current set $P$.
\medskip

Answering nearest-neighbor queries for the weighted point set $P$ can be done by performing point location
in the \emph{additively weighted Voronoi diagram} of $P$. The additively weighted
Voronoi diagram of~$P$, denoted by $\awvd(P)$, is the subdivision of
the plane into regions such that the region of a point $p_k\in P$ consists of
those points $q\in \Reals^2$ for which $p_k$ is the nearest neighbor of~$q$ if we consider
additively weighted distances.
\footnote{Sometimes additive weighted Voronoi diagrams
are defined based on a distance function that subtracts a positive weight from the Euclidean distance.
It is easy to see that all results carry over the case where we add weights, because
we can transform the latter case to the former by subtracting the same sufficiently large value
from all weights to make them negative.}
The diagram $\awvd(P)$ consists of at most $n$
regions---at most, because some points may define an empty region---and the boundaries between the regions
consist of hyperbolic arcs. The total complexity of the diagram is $\Oh(n)$ and
it can be computed in $\Oh(n\log n)$ time~\cite{Fortune87}.
Moreover, point location in a planar subdivision of complexity $\Oh(n)$
can be done in $\Oh(\log n)$ time with a data structure that uses $\Oh(n)$ storage
and $\Oh(n\log n)$ preprocessing~\cite{Snoeyink04PL}. Thus nearest-neighbor queries in $P$
under the additively weighted distance function can be done in $\Oh(\log n)$
after $\Oh(n\log n)$ preprocessing.

We first briefly review the logarithmic method. It makes a static data structure~$\ds$
semi-dynamic, as follows. Let~$n$ be the number of objects---weighted points in our case---in
the set $S$ currently stored in the data structure, and let $a_t \in \{0,1\}$ be such that
$n = \sum_{t=0}^{\floor{\log n}} a_t 2^t$.
The logarithmic method maintains, for each $t$ with $a_t =1$, a static data structure~$\ds^{(t)}$
on a subset $S^{(t)}\subseteq $ of size~$2^t$, where
the subsets $S^{(t)}$ form a partition of~$S$.
A query on the set $S$ can now be answered by querying each of the data structures~$\ds^{(t)}$,
and computing the final answer to the query from the $\Oh(\log n)$ sub-answers.
(The use of the logarithmic method thus requires the query problem to be such
that the answer to a query on the whole set~$S$ can be easily computed from the
answers on the subsets~$S^{(t)}$.)
To insert a new object~$o$, one first finds the smallest $t^*$ such that $a_{t^*}=0$,
where the $a_t$ are defined with respect to the size of $S$ before the insertion.
Then all structures $\ds^{(0)},\ldots,\ds^{(t^*-1)}$ are destroyed, and a new structure
$\ds^{(t^*)}$ on the set~$S^{(0)} \cup \cdots \cup S^{(t^*-1)} \cup \{o\}$ is constructed.
The amortized insertion time is $\Oh(\sum_{t=0}^{\floor{\log n}} B(2^t)/2^t)$, where
$B(2^t)$ denotes the time needed to construct a data structure on a set of size~$2^t$.

In our case each $\ds^{(t)}$ is a point-location structure for the additively weighted
Voronoi diagram $\awvd(P^{(t)})$ on a subset $P^{(t)}\subseteq P$.
Note that we can easily find the overall nearest neighbor of a query point by
taking the nearest among the $\Oh(\log n)$ candidates found for the subsets~$P^{(t)}$.
Thus our structure has $\Oh(\log^2 n)$ query time. Since a substructure $\ds^{(t)}$ can be built
in $\Oh(|P^{(t)}|\log |P^{(t)}|)$ time, the amortized time for an insertion is
$\Oh(\log^2 n)$.

It remains to deal with bulk updates, where we want to increase the
weight of each of the points in our data structure by a given value~$\Delta$.
With the logarithmic method this is quite easy. We simply store a
\emph{correction term} $\Delta_t$ for each~$\ds^{(t)}$, which indicates
that the weight of each point in $P^{(t)}$ should be increased
by~$\Delta_t$. A bulk update with value~$\Delta$ can then be performed in $\Oh(\log n)$ time
by adding~$\Delta$ to each of the correction terms~$\Delta_t$.
Note that we can still answer queries correctly. Indeed, $\awvd(P^{(t)})$
does not change when we add the same value~$\Delta$ to all weights in~$P^{(t)}$.
Hence, we just have to make sure that when we compare the
candidates found for the subsets~$P^{(t)}$, we increase their weighted
distances by the relevant correction term. Thus a query still takes
$\Oh(\log^2 n)$ time. Insertions can still be done in $\Oh(\log^2 n)$ amortized time
as well; we only need to make sure that when we collect the points in the substructures
$\ds^{(0)},\ldots,\ds^{(t^*-1)}$ to be destroyed, we add the correct terms to their weights
before we construct the new structure~$\ds^{(t)}$. Hence, we obtain the claimed
query time, bulk-update time, and insertion time, thus finishing the proof of Theorem~\ref{th:regular}.


\section{Planar bottleneck pyramidal TSP}
\label{se:bottleneck-decision}
Below we consider the bottleneck version of the pyramidal \tsp problem.
The goal is to find a pyramidal tour for an ordered set $P:=\{p_1,\ldots,p_n\}$
of points in the plane such that the length of the bottleneck edge (that is, the
longest edge) is minimized. We start by giving an $\Oh(n\log n)$
algorithm for decision version of the problem,
where we are given a value~$B$ and the question is whether there is
a pyramidal tour whose bottleneck edge has length at most~$B$.
Next we show that this is optimal by presenting an $\Omega(n\log n)$ lower bound
in the algebraic computation-tree model. Finally, we show how to solve the
optimization version of the problem.

\subsection{An algorithm for the decision problem}
\label{subse:bottleneck-decision-alg}
The decision problem can be solved by dynamic programming, using a 2-dimensional table $A[1..n,1..n]$,
where (for $1\leq j<i<n$) we have $A[i,j]=\true$ if there is an $(i,j)$-partial tour of cost at most~$B$
and $A[i,j]=\false$ otherwise. The dynamic program can compute the entries $A[i,j]$
row by row with the recursive formula
\[
A[i+1,j] \ = \ \left\{ \begin{array}{ll}
                    A[i,j] \wedge (|p_i p_{i+1}| \leq B) & \mbox{if $1\leq j < i$} \\[2mm]
                    \bigvee_{1\leq k<i} \left(A[i,k] \wedge |p_kp_{i+1}| \leq B \right) & \mbox{if $j=i$}
                    \end{array}
             \right.
\]
where $A[2,1]=\true$ if $|p_1 p_2| \leq B$.
As before, we speed up the computation by using the relation between consecutive rows in the
table---for $j<i$ the entries $A[i+1,j]$ are all
equal to $A[i,j]$ when $|p_i p_{i+1}| \leq B$, and they are all $\false$ otherwise---and
by using appropriate geometric data structures.
\medskip

Instead of computing the entries of the matrix~$A$, we will maintain a list
$\mylist$ that contains, for the current value of~$i$, all points $p_j$ with $1\leq j<i$
such that there is an $(i,j)$-partial tour of cost at most~$B$. In other words, $\mylist$
contains all $p_j$ such that $A[i,j]=\true$.
We initialize $\mylist$ as an empty list, and then go over 
the points $p_2,\ldots,p_n$ in order. To handle $p_{i+1}$ we check if
$|p_{i}p_{i+1}|\leq B$, and we check if $\mylist$ currently contains a point
$p_k$ such that $|p_kp_{i+1}| \leq B$. If both conditions are satisfied
we add $p_i$ to $\mylist$, if only the first condition is satisfied
we keep $\mylist$ as it is, if only the second condition is satisfied
we first empty $\mylist$ and then add $p_i$ to it, and if neither condition
is satisfied then we empty $\mylist$. After having handled $p_{n-1}$
we check if $|p_{n-1} p_n|\leq B$ and if $\mylist$ contains a point
$p_k$ such that $|p_kp_{n}| \leq B$. If both conditions are satisfied
then a pyramidal tour whose bottleneck length is at most $B$ exists,
otherwise it does not.

Since a point is added to $\mylist$ only once, the total number of updates
to $\mylist$ is $\Oh(n)$.
Checking the first condition obviously takes $\Oh(1)$ time, so it remains to describe
how to check the second condition efficiently. To this end we
maintain a data structure on the points in $\mylist$ that supports
two operations:
\begin{itemize}
\item \emph{query} the structure with a point~$p_{i+1}$ to decide
      if it stores a point $p_k$ with $|p_k p_{i+1}|\leq B$;
\item \emph{insert} a new point $p_{i}$ into the structure.
\end{itemize}
These operations can be performed by a semi-dynamic data structure for
nearest-neighbor queries, similar to the one described earlier, which
has $\Oh(\log^2 n)$ query time and $\Oh(\log^2 n)$ amortized insertion time.
Below we describe a faster data structure. The data structure is based on the following
observation. Let $D_k$ be the disk of radius~$B$ centered at the point~$p_k$,
and let
\[
\D := \{ D_k : \mbox{ $p_k$ is a point in $\mylist$}\}.
\]
Then $\mylist$ contains a point $p_k$ with $|p_k q| \leq B$
if and only if $q\in\union(\D)$, where $\union(\D)$ denotes
the union of the disks in $\D$. Theorem~\ref{th:pl-in-disk-union}
below states that point-location queries in the union of
a set of congruent disks can be done in $\Oh(\log n)$ time and
with $\Oh(\log n)$ amortized update time, leading to the following result.
\begin{untheorem}
Let $P$ be an ordered set of $n$ points in the plane, and let $B>0$ be a given parameter.
Then we can decide in $\Oh(n\log n)$ time and using $\Oh(n)$ storage
if $P$ admits a pyramidal tour whose longest edge has length at most~$B$.
\end{untheorem}

\paragraph{A semi-dynamic point-location data structure for the union of congruent disks.}
Let $\D$ be a set of congruent disks in the plane. We wish to maintain a data structure
on $\D$ that allows us to decide if a query point $q$ lies inside $\union(\D)$. The
data structure should also allow insertions into the set~$\D$. With a slight abuse
of notation, we will use $n$ to denote the number of disks in the (current) set~$\D$.
We will assume we have the floor function available; it is not hard to avoid the floor
function, but using it simplifies the presentation.
It will also be convenient to assume that the disks in $\D$ all have radius~$\sqrt{2}$,
which can be ensured by appropriate scaling.
\medskip

Consider the integer grid $G$. Note that the diameter of the grid cells is $\sqrt{2}$,
so any cell containing the center of some disk $D_i$ is completely covered by~$D_i$.
We say that a grid cell\footnote{To assign each point to a unique active cell,
we assume the cells in $G$ are closed on the left and bottom, and open on
the right and top. Thus the cells in $G$ are of the form $[x,x+1)\times[y,y+1)$
for integers $x,y$.}~$C$ is \emph{active} if it contains
the center of a disk $D_i \in \D$,
and we say that a vertical strip $[x,x+1)\times(-\infty,\infty)$ is
active if it contains an active grid cell.
Our data structure for point location in $\union(\D)$ maintains the active strips in a
balanced search tree on their $x$-order, and for each active strip~$\Sigma$ it maintains the active
cells within~$\Sigma$ in a balanced search tree on their $y$-order. (These search trees could also be
replaced by a hash table.) For each active cell~$C$ we maintain four partial unions,
as explained next.

Let $\D(C)\subseteq\D$ be the set of disks whose center lies in~$C$.
Let $\ell_{\mytop}(C)$, $\ell_{\mybot}(C)$, $\ell_{\myleft}(C)$, and $\ell_{\myright}(C)$
denote the lines containing, respectively, the top, bottom, left, and right edge of~$C$.
Finally, define $U_{\mytop}(C)$, $U_{\mybot}(C)$, $U_{\myleft}(C)$, and $U_{\myright}(C)$
to be the parts of $\union(\D(C))$ lying, respectively, above $\ell_{\mytop}(C)$,
below $\ell_{\mybot}(C)$, to the left of $\ell_{\myleft}(C)$, and to the
right of~$\ell_{\myright}(C)$. Next we explain how we store and maintain the
partial union $U_{\mytop}(C)$; the other three partial unions are stored
and maintained in a similar manner.
\medskip

Let $p_i$ denote the center of the disk $D_i\in \D(C)$. Because the centers~$p_i$ all lie
inside $C$, they all lie below the line ~$\ell_{\mytop}(C)$. Hence, the partial union
$U_{\mytop}(C)$ is $x$-monotone. Furthermore, each component of $U_{\mytop}(C)$
is bounded from below by a portion of the line~$\ell_{\mytop}(C)$ and from above
by circular arcs that are portions of the boundaries of the disks $D_i\in \D(C)$.
The key to efficiently maintaining $U_{\mytop}(C)$ is the following lemma.
\begin{lemma}\label{le:consistent-order}
Each disk $D_i\in \D(C)$ contributes at most one arc to $\bd U_{\mytop}(C)$. Moreover,
the arc contributed by a disk $D_i$ lies to the left of the arc contributed by
a disk $D_j$ if and only if $p_i$ lies the left of~$p_j$.
\end{lemma}
\begin{proof}
Define $\gamma_i$ to be the part of $D_i$'s boundary
above the line~$\ell_{\mytop}(C)$. Any other disk $D_j\in \D(C)$ that covers a part of $\gamma_i$ must
contain an endpoint of~$\gamma_i$. Indeed, if $\gamma_j$ would intersect $\gamma_i$ twice
above~$\ell_{\mytop}(C)$ then, since the centers of the disks $D_i$ and $D_j$ lie
below~$\ell_{\mytop}(C)$, the curvature of $\gamma_j$ would be larger than the curvature
of $\gamma_i$, contradicting the fact that all disks have equal radius. Hence, each
disk $D_i$ can contribute at most one arc to $\bd U_{\mytop}(C)$, as claimed.
\begin{figure}
\begin{center}
\includegraphics{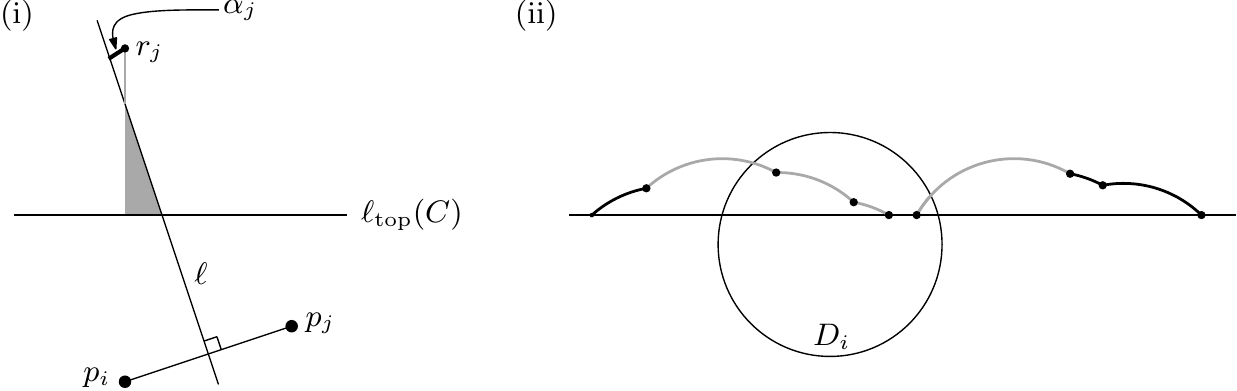}
\end{center}
\caption{(i) Illustration for the proof of Lemma~\ref{le:consistent-order}.
        (ii) The addition of $D_i$ causes several arcs to disappear from
             $\bd U_{\mytop}(C)$ and two arcs to be shortened. All these arcs
             (indicated in gray) are consecutive in the left-to-right order.}
\label{fi:consistent-order}
\end{figure}

Now consider an arc $\alpha_i\subseteq \gamma_i$ contributed by $D_i$ and an
arc~$\alpha_j\subseteq \gamma_j$ contributed by $D_j$. Assume without loss of
generality that $p_i$ lies to the left of~$p_j$. Furthermore, assume $p_i$
lies below~$p_j$, as in Fig.~\ref{fi:consistent-order}(i); a similar argument applies
when $p_i$ lies above~$p_j$.
Now suppose for a contradiction that $\alpha_i$ lies to the right of~$\alpha_j$.
Let $\ell$ be the perpendicular bisector of the segment~$p_i p_j$. Because the disks
have equal radius, $\alpha_i$ must lie to the left of $\ell$ and $\alpha_j$ must lie to the right of~$\ell$.
Hence, if $\alpha_i$ to lies to the right of $\alpha_j$, then it must lie in the
triangular region bounded by $\ell_{\mytop}(C)$, and $\ell$, and the vertical line
through~$r_j$. (In Fig.~\ref{fi:consistent-order}(i) this triangle is shown shaded.)
But this region is completely contained in $D_j$ since any point in it is closer to $p_j$
than $r_j$ is. Hence, we have a contradiction with the fact that
$\alpha_i$ is an arc of $\bd U_{\mytop}(C)$.
\end{proof}
Lemma~\ref{le:consistent-order} gives us an easy way to store and maintain~$U_{\mytop}(C)$.
We simply store the arcs comprising $\bd U_{\mytop}(C)$ in $x$-order in a balanced
search tree~$\tree_{\mytop}(C)$. This takes $\Oh(n)$ storage and allows us to decide
in $\Oh(\log n)$ time for a query point~$q$ if $q\in U_{\mytop}(C)$.

Now suppose we want to insert a new disk $D_i$ into $\D(C)$.
As observed, $\bd D_i$ contributes at most one new arc to $\bd U_{\mytop}(C)$.
The addition of this new arc means we have to remove some existing arcs.
More precisely, if $\bd D_i$ contributes a new arc, then we have to shorten
two existing arcs and possibly remove
one or more other arcs; see Fig.~\ref{fi:consistent-order}(ii).
Once we know which existing arcs are affected, the update can be done
in $\Oh((k+1)\log n)$ time, where $k$ is the number of disappearing arcs.
Since each arc is removed at most once, this given an amortized insertion time of $\Oh(\log n)$.
It remains to describe how to check whether $D_i$ actually contributes a new arc
and, if so, which existing arcs are affected. To this end we search in
$\tree_{\mytop}(C)$ for an affected arc, that is, for an arc
that is completely or partially covered by~$D_i$. When we have such an arc,
we can easily find all other affected arcs, because these arcs are neighbors
in the left-to-right ordering. It remains to describe how to search for an
affected arc.

Consider the arc $\alpha_j$ (contributed by some disk $D_j$) stored at the
root of~$\tree_{\mytop}(C)$. If $D_i$ covers (a part of) $\alpha_j$
then we have found an affected arc. Otherwise, if $p_i$ lies to the left of $p_j$
then we recursively search in the left subtree of the root, and else
we recursively search in the right subtree. This continues until we either
find an affected arc, or we reach a leaf. In the latter case $D_i$ does
not contribute a new arc to~$\bd U_{\mytop}(C)$. The correctness of this procedure
is guaranteed by Lemma~\ref{le:consistent-order}.
We can conclude the following lemma.
\begin{lemma}\label{le:envelope-union}
We can maintain $U_{\mytop}(C)$ in a data
structure using $\Oh(n)$ storage such that we can decide in $\Oh(\log n)$ time for a
query point~$q$ if $q\in U_{\mytop}(C)$. The data structure can be maintained under
insertions in $\Oh(\log n)$ amortized time.
\end{lemma}
To summarize, our point-location data structure for $\union(\D)$ consists of the following
components.
\begin{itemize}
\item A balanced search tree $\tree$ storing the active strips sorted on their $x$-order,
      and for each active strip $\Sigma$ a balanced search tree $\tree_{\Sigma}$ on the active cells inside that
      strip, sorted on $y$-order.
\item  For each active cell~$C$, the partial union $U_{\mytop}(C)$ is stored in
       the data structure of Lemma~\ref{le:envelope-union}. The other three partial unions
       $U_{\mybot}(C)$, $U_{\myleft}(C)$, and $U_{\myright}(C)$ are stored in
      similar data structures.
\end{itemize}
To answer a query we first determine the grid cell $C_q$ containing the query point~$q$.
If $C_q$ is active, we know that $q\in \union(\D)$. Otherwise we determine the
\emph{relevant grid cells} for~$q$, that is, the active cells $C$ whose distance to
$q$ is at most~$\sqrt{2}$; these are the only cells for which $\D(C)$ can contain a disk
$D_i$ such that $q\in D_i$. Note that there are only $\Oh(1)$ such cells and that
they can be found in $\Oh(\log n)$ time using the tree $\tree$ and the trees $\tree_{\Sigma}$.
For each relevant cell~$C$, we then query the appropriate partial union; for example,
if $q$ lies above $\ell_{\mytop}(C)$ we query $\tree_{\mytop}(C)$.
Now $q$ lies in $\union(\D)$ if and only if $q$ lies in at least
of one these partial unions.

Inserting a new disk $D_i$ is done as follows. First we determine the grid cell $C_q$
containing the center $p_i$ of $D_i$. If $C_q$ is not yet active, we insert $C_q$
into our structure (when necessary first creating a new active strip). Next we insert~$D_i$
into each of the four partial unions stored for~$C_q$. By Lemma~\ref{le:envelope-union}
the whole procedure takes $\Oh(\log n)$ amortized time.
{
\renewcommand{\thetheorem}{\ref{th:pl-in-disk-union}}
\begin{theorem}
We can maintain a collection $\D$ of $n$ congruent disks in a data structure such that
we can decide in $\Oh(\log n)$ time if a query point $q$ lies in $\union(\D)$.
The data structure uses $\Oh(n)$ storage and a new disk can be inserted into $\D$ in $\Oh(\log n)$
amortized time.
\end{theorem}
\addtocounter{theorem}{-1}
}


\subsection{A lower bound for the decision problem}
\label{subse:pyramidal-lb}
Below we show an $\Omega (n \log n)$ time lower bound for the decision version of the bottleneck pyramidal \tsp
in the Euclidean plane, in the algebraic computation-tree model.
The reduction even applies to the bitonic setting where the points are ordered from left to right.
This bound matches the upper bound in Theorem~\ref{th:bottleneck-decision}.

\begin{theorem}\label{th:bottleneck-decision-lb1-pyramidal}
The bottleneck pyramidal TSP on $n$ points in the Euclidean plane has a lower bound of $\Omega (n \log n)$ in
the algebraic computation-tree model.
\end{theorem}
\begin{proof}
We prove the lower bound by a reduction from \emph{set disjointness} for integer sets $U = \{ u_1, \ldots u_n \}$ and $V = \{ v_1, \ldots v_n \}$, for which the lower bound is known~\cite{y-lbact-91}. Without loss of generality we may assume that all integers are positive. We need to construct an ordered set of points $P$ and choose a bound $B>0$ such that $U \cap V \neq \emptyset$  if and only if $P$ admits a pyramidal tour whose longest edge has length at most~$B$.

Let $M := \max U \cup V$ and $B := M+1$. We define $P = \{ p_1, \ldots, p_{2n+2} \}$ with $p_i = (0, u_i)$ for $1 \leq i \leq n$, $p_{n+1} = (0,B)$, $p_{n+2} = (B,B)$, and  $p_{n+2+i} = (B, v_i)$ for $1 \leq i \leq n$ (see Fig~\ref{fig:lowerbound}).

First assume $U \cap V \neq \emptyset$ and $u_{i'} = v_{j'}$. The pyramidal tour that first visits all $p_i$ in order except $p_{i'}$ and $p_{n+2+j'}$ and finally $p_{n+2+j'}$ and $p_{i'}$ has only edges of length at most $B$. Conversely, assume there is a pyramidal tour whose longest edge has length at most $B$. Any tour on $P$ needs to move from the line $x=0$ to the line $x=B$ and back. Since the corresponding edges have length at most $B$, they need to connect points in $P$ with the same $y$ coordinate.
This implies that there has to be a $y \neq B$ such that $(0,y)$ and $(B,y)$ are in $P$, which in turn implies that $U$ and $V$ are not disjoint.
\end{proof}

\begin{figure}[t]
	\centering
  \includegraphics{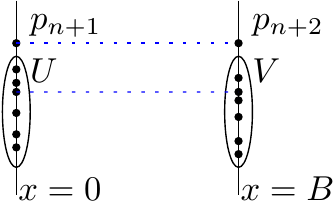}
	\caption{Lower bound construction: A TSP tour needs to move between the lines $x=0$ and $x=B$ at least twice, which is possible with a bottleneck of B exactly if $U$ and $V$ have at least one element in common.}
	\label{fig:lowerbound}
\end{figure}

The construction above does not immediately work for the bottleneck bitonic TSP, since it uses points with the same $x$-coordinate. However, we can slightly perturb the points to obtain unique $x$-coordinates.

\begin{theorem}\label{th:bottleneck-decision-lb1-bitonic}
The bottleneck bitonic TSP on $n$ points in the Euclidean plane has a lower bound of $\Omega (n \log n)$ in
the algebraic computation tree model.
\end{theorem}
\begin{proof}
Let $\Delta := \frac{1}{4B(n+1)}$. We use the same construction as in the previous proof, except that we slightly change the $x$-coordinates of the points in $P$. Concretely, we set the $x$-coordinate of $p_i$ to $i\Delta$ for $i \leq n+1$ and to $B-(i-n+1)\Delta$ for $i > n+1$. Since $(n+1)\Delta < 1$, the distance between any $p_i$ and $p_j$ with $1 \leq i,j, \leq n+1$ still is less than $B = M+1$, and likewise for $n+2 \leq i,j \leq 2n+2$. Thus, if $U \cap V \neq \emptyset$, we obtain the same tour as in the previous proof with edges of length at most $B$. Conversely, if there is a bitonic tour with edges of length at most B, we need to check that edges crossing the line $x = B/2$ do not connect points with different $y$-coordinates. Suppose there is such an edge, then its length would be at least
\[
\sqrt{1 + (B-2(n+1)\Delta)^2} = \sqrt{1 + (B-\frac{1}{2B})^2} > \sqrt{1 + B^2 - 1} = B,
\]
a contradiction.
\end{proof}

\subsection{An algorithm for the optimization problem}
\label{subse:bottleneck}
In the optimization version of the bottleneck pyramidal \TSP problem the goal
is to minimize the length of the bottleneck edge, that is, the length of the longest
edge in the tour.
%
\medskip

The standard dynamic-programming solution for the optimization version of the pyramidal
bottleneck \TSP uses a table~$A[1..n,1..n]$
where $A[i,j]$ is defined as the minimum
value for $B$ such that there is an $(i,j)$-partial tour of cost at most~$B$.
We have
\[
A[i+1,j] \ = \ \left\{ \begin{array}{ll}
                    \max(A[i,j], |p_i p_{i+1}|) & \mbox{if $1\leq j < i$} \\[2mm]
                    \min_{1\leq k<i} \max(A[i,k], |p_k p_{i+1}|) & \mbox{if $j=i$}
                    \end{array}
             \right.
\]
where $A[2,1]=|p_1 p_2|$. Our strategy to speed up the dynamic-programming
algorithm is similar to the strategy for the non-bottleneck version in Section~\ref{se:pyramidal}: we view the
values $A[k,i]$ as the weight of the point~$p_k$ in the $i$-th iteration of
the algorithm, and we maintain the points with their weights in a suitable
data structure. This time the data structure needs to support the following operations:
\begin{itemize}
\item perform a \emph{query} with point $q$, which reports the
			value $\min_{p_k} \max(w_k, |p_k q|)$, where the min is over all points $p_k$
      currently in the data structure;
\item perform a \emph{bulk update} of the weights, which
      sets $w_j := \max(w_j, B)$ for each point $p_j$ currently in the data structure,
      for a given value~$B$;
\item \emph{insert} a new point $p_{i}$ with given weight~$w_{i}$ into the data structure.
\end{itemize}
Below we describe a data structure supporting these operations with $\Oh(\log^3 n)$
query time, $\Oh(\log^3 n)$ amortized insertion time and $\Oh(\log n)$ time for bulk updates.
The structure uses $\Oh(n\log n)$ storage, leading to the following theorem.
{
\renewcommand{\thetheorem}{\ref{th:bottleneck}}
\begin{theorem}
Let $P$ be an ordered set of $n$ points in the plane.
Then we can compute a pyramidal tour whose bottleneck edge has minimum
length in $O (n\log^3 n)$ time and using $\Oh(n\log n)$ storage.
\end{theorem}
\addtocounter{theorem}{-1}
}

\paragraph{The data structure.}
Below we describe a data structure that supports queries and bulk updates. To support
insertions, we then apply the logarithmic method.
With a slight abuse of notation, we let $P:=\{p_1,\ldots,p_n\}$ denote the weighted point set
stored in the data structure.
Let $W$ be the (multi-)set of the weights of the points in~$P$. Our data structure
is defined as follows.
\begin{itemize}
\item The main tree is a balanced search tree $\tree$ whose leaves store the weights from~$W$,
      together with the corresponding points. For a node $\node$ in $\tree$, let $P(\node)$
      denote the set of points stored in the subtree rooted at~$\node$.
      We maintain the following information at~$\node$.
      \begin{itemize}
      \item Let $D(p_j,w_j)$ be the disk centered at the point $p_j$ and of radius~$w_j$.
            We store the union of the set $\{D(p_j,w_j) : p_j \in P(\node) \}$,
            preprocessed for point location. Here the weights $w_j$ refer to the weights
            at the time the data structure was constructed; after a bulk
            update the union~$U(\node)$ is not changed. We denote this union by~$U(\node)$.
      \item The Voronoi diagram $\vd(P(\node))$ of the point set $P(\node)$ (using the
            normal Euclidean distances), preprocessed for point location.
      \item The maximum weight stored in the subtree rooted at~$\node$.
      \end{itemize}
\item We maintain $B_{\max}$, the maximum value of any of the bulk updates executed since
      the construction of the data structure.
\end{itemize}
Since the unions and Voronoi diagrams stored at each node (and their point-location
data structures) use linear storage~\cite{Fortune87,Snoeyink04PL}, the overall amount of storage
of our data structure is $\Oh(n\log n)$.
The idea behind the query procedure, which will be described below, is the following lemma.
Recall that a query with a point~$q$ should return the value
$B(q) := \min_{p_j\in P} \max(w_j, |p_j q|)$.
\begin{lemma}\label{le:w1andw2}
Let $B_1 := \min \{ w_j : p_j\in P \mbox{ and } |p_j q| \leq w_j \}$ and
$B_2 := \min \{ |p_j q| : p_j\in P \mbox{ and } w_j < B_1 \}$. Then
$B(q) = \min (B_1,B_2)$.
\end{lemma}
\begin{proof}
Define $P_1 := \{ p_j \in P: w_j \geq |p_j q| \}$ and
$P_2 := \{ p_j \in P: w_j < |p_j q| \}$. Note that
$B_1 = \min_{p_j\in P_1} w_j$ and define $B'_2 := \min_{p_j\in P_2} |p_j q|$.
Clearly $B(q) = \min (B_1,B'_2)$.

Because of the definition of $B_1$, we have $\{ p_j : w_j < B_1 \} \subseteq P_2$.
Hence, $B_2 \geq B'_2$. Furthermore, if $B_2 > B'_2$ then
the point $p_j\in P_2$ minimizing $|p_j q|$ has $w_j\geq B_1$, and so
$\min(B_1,B'_2) = B_1 = \min(B_1,B_2)$ in this case. Trivially
$\min(B_1,B'_2) = \min(B_1,B_2)$ also holds when $B_2 = B'_2$.
Hence, $B(q) = \min (B_1,B_2)$, as claimed.
\end{proof}
We now describe how to perform the three operations on $\tree$.

\emph{Queries.}
To answer a query we first compute the nearest neighbor, $p_k$, of $q$ in $P$.
This can be done in $\Oh(\log n)$ time by locating $q$ in~$\vd(\myroot(\tree))$,
since $P(\myroot(\tree))=P$.
If $|p_k q| \leq B_{\max}$ we can immediately conclude that $B(q) = B_{\max}$.
Otherwise we answer the query by computing $B_1$ and $B_2$,
and then returning~$\min(B_1,B_2)$. Next we explain how to compute~$B_1$ and $B_2$.
Note that when we have to do so, we know that $|p_j q| > B_{\max}$ for all
$p_j\in P$. This implies that for any given node~$\node$ we can decide
whether there is a point $p_j\in P(\node)$ with $|p_j q|\leq w_j$
by checking if $q\in U(\node)$---the bulk updates we
have performed since constructing~$U(\node)$ do not affect the outcome.

We can compute $B_1$ as follows.
We start by checking  if $q\in U(\myroot(\tree))$.
If this is not the case then $\{ p_j\in P : |p_j q| \leq w_j \} =\emptyset$ and
so we set $B_1 := \infty$. Otherwise we walk down the tree, as follows.
Suppose we are at a non-leaf node $\node$. Let $\mu$ be the left child of $\node$.
If  $q\in U(\mu)$ then we descend to the left child of~$\node$ (that is,
we set $\node := \mu$) and otherwise we proceed to the right child.
Since the points of $P$ are stored in the leaves of $\tree$ in order
of their weights, the search will end in the leaf storing the point~$p_{j*}$
with the smallest weight among the nodes $p_j$ with $|p_j q| \leq w_j$.
Thus we set $B_1 := w_{j^*}$.

Next we need to compute~$B_2$. As observed earlier, when we have to compute
$B_1$ and $B_2$ we know that $|p_j q| > B_{\max}$ for all $p_j\in P$.
Hence, $B_1 > B_{\max}$. This implies that whether or not a point $p_j$
satisfies~$w_j < B_1$ is not affected by the bulk updates done so far---we
can use the weights at the time $\tree$ was constructed to find the points
$p_j$ satisfying~$w_j < B_1$. To compute $B_2$ we now identify a collection of
$\Oh(\log n)$ nodes $\node$ such that the sets $P(\node)$ contain exactly
the points $p_j$ with~$w_j < B_1$. This can be done by searching with $B_1$
in $\tree$.
At each of these nodes we compute
$\min_{p_j\in P(\node)} |p_j q|$ by point location in $\vd(P(\node))$,
and we set $B_2$ to be the minimum of the $\Oh(\log n)$ values computed
in this manner.

Both $B_1$ and $B_2$ are computed in $\Oh(\log^2 n)$ time---indeed, for both
we spend $\Oh(\log n)$ at each node along a path in $\tree$---so
the total query time (before applying the logarithmic method) is $\Oh(\log^2 n)$ time.

\emph{Bulk updates.}
A bulk update with value~$B$ is performed in $\Oh(1)$ time by setting
$B_{\max} := \max(B_{\max},B)$; no other action is needed.

\emph{Insertions.}
Insertions are handled using the logarithmic method. This increases the time for queries
and bulk updates to $\Oh(\log^3 n)$ and $\Oh(\log n)$, respectively. The amortized time
for insertions is $\Oh( T_B(n) \log n)$, where $T_B(n)$ is the time needed to build
a static structure on $n$ points. This can be done bottom-up in $\Oh(n\log^2 n)$ time:
At each node $\node$ we can construct the point-location data structure on the union~$U(\node)$
in $\Oh(|P(\node)| \log |P(\node)|)$ time~\cite{Snoeyink04PL}, and we can construct the Voronoi diagram in
the same amount of time~\cite{Fortune87}. (Before we can construct the point-location data structure
we first need to construct~$U(\node)$, but this can be done in $\Oh(|P(\node)| \log |P(\node)|)$
time by merging the unions from the two children of~$\node$.) We conclude that
the amortized time for insertions is~$\Oh(\log^3 n)$.


\section{On truly subcubic algorithms for 3-OPT: Missing proof}
\label{se:appendix-reduction}
\begin{lemma} \label{lemma:threeopt:to:negtriangle}
\ThreeOPTDetection can be reduced to \NegativeTriangle in time~$\Oh(n^2)$ while increasing the size of the graph and the largest weight by a constant factor.
\end{lemma}
\begin{proof}
Consider an instance of \ThreeOPTImprovement, which is given by a complete graph~$G$ together with a tour~$T$ in~$G$ and a symmetric distance function~$d$. Number the vertices of~$G$ as~$v_1, \ldots, v_n$ in the order of~$T$. Let~$M$ be the largest absolute value of an edge weight. To simplify the notation that we will need, we first deal with two simple cases. In~$\Oh(n^2)$ time we check whether there is an improving \Twomove in~$G$. If so, we simply output a constant-size \yes-instance as the output of the reduction. In the remainder it suffices to look for a \Threemove that removes three edges and replaces them by three different edges. Secondly, we test whether there is an improving \Threemove where two of the removed edges share an endpoint. This can be done in~$\Oh(n^2)$ time: there are~$n$ possibilities for the shared endpoint, which determines the first two edges to leave the tour, and~$n$ options for the third edge that leaves the tour. Each option can be handled in constant time. In the remainder it therefore suffices to produce an input of \NegativeTriangle whose answer is \yes if and only if there is a \Threemove that removes three distinct edges that do not share any endpoint, and replaces them by three different edges. In the remainder of this proof, we refer to such a \Threemove as a \emph{proper \Threemove}.

To reduce the problem of finding a proper \Threemove to that of finding a negative-weighted triangle, we consider the different ways in which the three paths that are obtained from~$T$ by removing three edges, can be connected back into a Hamiltonian cycle of the graph by replacing them with different edges. Consider the graph on vertices~$a_0, a_1, b_0, b_1, c_0, c_1$ with edges~$\{a_1, b_0\}, \{b_1,c_0\}$, and~$\{c_0,a_1\}$, which represents an abstract tour on these vertices from which edges~$\{a_0, a_1\}$,~$\{b_0, b_1\}$, and~$\{c_0, c_1\}$ have been removed; see Figure~\ref{fi:characteristics}. The removals result in three gaps: the~$a$-gap (between~$a_0$ and~$a_1$), the~$b$-gap, and the $c$-gap. Each set of~$3$ edges that completes this graph into a cycle without inserting any of the removed edges~$\{a_0, a_1\}$,~$\{b_0, b_1\}$, or~$\{c_0, c_1\}$, can be characterized by~$6$ bits~$\ell(ab), r(ab), \linebreak[0] \ell(ac), r(ac), \ell(bc), r(bc) \in \{0,1\}$ such that the edges completing the graph into a cycle are $\{a_{\ell(ab)}, b_{r(ab)}\}, \linebreak[1] \{a_{\ell(ac)}, c_{r(ac)}\}$, and~$\{b_{\ell(bc)}, c_{r(bc)}\}$. The bit~$\ell(ab)$ specifies, for example, whether the edge connecting the $a$-gap to the $b$-gap attaches to the left side of the $a$-gap ($\ell(ab)=0$), or to the right side of the $a$-gap. The bit~$r(ab)$ specifies whether the connection between the $a$-gap and $b$-gap attaches to the left or right side of the $b$-gap, and so on. For each set of~$3$ edges that completes the graph into a cycle without re-inserting a removed edge, make a weighted 3-partite connected component with~$3n$ vertices~$\{x_i, y_i, z_i \mid i \in [n]\}$ and edge weights defined as follows:

\begin{figure}
\begin{center}
\includegraphics{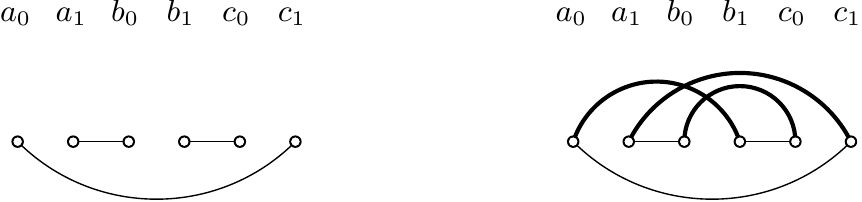}
\end{center}
\caption{Left: the $6$-vertex template graph with the $a$-gap, $b$-gap, and~$c$-gap. Right: the thick edges give one possibility for completing the graph into a cycle, with characteristic~$\ell(ab) = 0, r(ab) = 1, \ell(ac) = 1, r(ac) = 1, \ell(bc) = 0, r(bc) = 0$.}
\label{fi:characteristics}
\end{figure}

\begin{itemize}
	\item $w(\{x_i, y_j\}) = d(v_{i+\ell(ab)}, v_{j + r(ab)}) - d(v_i, v_{i+1})$ for~$1 \leq i < j - 1 \leq n$;
	\item $w(\{x_i, z_k\}) = d(v_{i+\ell(ac)}, v_{j + r(ac)}) - d(v_k, v_{k+1})$ for~$1 \leq i < k - 1 \leq n$;
	\item $w(\{y_j, z_k\}) = d(v_{j+\ell(bc)}, v_{k + r(bc)}) - d(v_j, v_{j+1})$ for~$1 \leq j < k - 1 \leq n$;
	\item the weight for the remaining pairs in the component is~$3M$.
\end{itemize}

Observe that by this definition, the weight of the triangle~$x_i, y_j, z_k$ for non-consecutive integers~$i<j<k$ is exactly the net weight change when removing the edges~$\{v_i, v_{i+1}\}$,~$\{v_j, v_{j+1}\}$, and~$\{v_k, v_{k+1}\}$ from the tour and replacing them as specified by the characteristic bits.

The weighted graph~$G'$ is the disjoint union of the connected components built for each characteristic. The weight of edges between different components is set to~$3M$.

\begin{claim}
The constructed instance of \NegativeTriangle has a triangle of negative edge-weight, if and only if the graph~$G$ allows an improving proper \Threemove.
\end{claim}
\begin{claimproof}
($\Rightarrow$) Assume that there exists an improving proper \Threemove for tour $T$ that removes the edges~$\{v_i, v_{i+1}\}, \{v_j, v_{j+1}\}$, and~$\{v_k, v_{k+1}\}$, producing tour~$T'$, and let~$i < j < k$. Since the endpoints of the removed edges are all distinct, we have~$i < j - 1$,~$j < k-1$, and therefore~$i < k-1$. Consider the~$3$ paths~$P_1,P_2,P_3$ that result from~$T$ by removing the three edges in their order along the original tour, such that~$P_1$ contains vertex~$v_1$. These paths are contained in tour~$T'$. To find the reconnection type corresponding to this move, replace each path~$P_i$ by a single edge. Relabeling the left and right endpoints of~$P_1, P_2, P_3$ to~$\{c_0, a_0\}$, and~$\{a_1, b_0\}$, and~$\{b_1, c_0\}$, respectively, we can now read off the reconnection type of the tour by seeing how the inserted edges of~$T'$ connect the relabeled vertices in the contracted graph. Consider the setting of the~$6$ bits~$\ell(ab), r(ab), \linebreak[0] \ell(ac), r(ac), \ell(bc), r(bc) \in \{0,1\}$ corresponding to this way of augmenting the six-vertex graph to a cycle. In the connected component corresponding to this choice of bits, the vertices~$\{x_i, y_j, z_k\}$ form a triangle. The total weight of this triangle is~$d(v_{i+\ell(ab)}, v_{j + r(ab)}) + d(v_{i+\ell(ac)}, v_{j + r(ac)}) + d(v_{j+\ell(bc)}, v_{k + r(bc)}) - d(v_i, v_{i+1}) - d(v_j, v_{j+1}) - d(v_k, v_{k+1})$. As the setting of the bits corresponds to the connection type of the \ThreeOPT move, this is exactly the sum of the weights of the newly introduced edges minus the weights of the removed edges. As the \ThreeOPT move gave a strict weight improvement, this value is negative and hence the vertices~$\{v_i, v_j, v_k\}$ from the specified component form a triangle of negative total weight.

($\Leftarrow$) Assume that the vertices $v_i,v_j,v_k$ span a triangle of negative edge-weight in $G'$. Since no weight is smaller than~$-M$, such a triangle cannot use a pair of weight~$3M$ and therefore consists of three vertices from a connected component that was added to~$G'$ on account of a specific reconnection pattern. Let~$i \leq j \leq k$. Since edges between vertices of the same letter also have weight~$3M$, as have edges going from larger indices to smaller ones, or between indices that differ at most one, we know that~$i < j - 1$ and~$j < k - 1$. Our weighting scheme ensures that removing the edges~$\{v_i, v_{i+1}\}, \{v_j, v_{j+1}\}, \{v_k, v_{k+1}\}$ and reconnecting the resulting pieces according to the reconnection pattern associated to the component, improves the weight of the tour by exactly the weight of triangle~$\{x_i, y_j, z_k\}$. Hence there is an improving \Threemove.
\end{claimproof}

The claim proves the correctness of the reduction. Since the number of characteristics is constant, the reduction can be done in~$\Oh(n^2)$ time and blows up the graph size and largest weight by only a constant factor.
\end{proof} 
\section{A fast $k$-OPT algorithm: missing proof} \label{se:appendix-kopt}

In this section we present an elementary reduction which shows that to find optimal an \kmove, it suffices to find a \kmove where the removed edges do not share any endpoints.

\begin{lemma} \label{lemma:subdividing}
For any~$k \geq 3$, an instance~$(G,T,d)$ of \kOPTOptimization can be reduced in time~$\Oh(n^2)$ to an instance~$(G',T',d')$, such that:
\begin{enumerate}
	\item $|V(G')| = 2 |V(G)|$,
	\item If the distances under~$d$ lie in the range~$[-M, \ldots, +M]$, then the distances under~$d'$ lie in the range~$[-2kM, \ldots, +M]$.
	\item Instance~$(G',T',d')$ has an optimal \kmove in which the removed edges do not share any endpoints.
	\item Given an optimal \kmove in~$(G',T',d')$, one can find an optimal \kmove in~$(G,T,d)$ in time~$\Oh(k)$.
\end{enumerate}
\end{lemma}
\begin{proof}
Consider an instance of \kOPTOptimization, which is given by a complete graph~$G$ together with a tour~$T$ in~$G$ and a symmetric distance function~$d$. The goal is to find a \kmove that improves tour~$T$ the most. Number the vertices of~$G$ as~$v_1, \ldots, v_n$ in the order of~$T$. Let~$M$ be the largest absolute value of an edge weight. Intuitively, the graph~$G'$ is obtained by subdividing all the edges on the current tour with a new vertex. One half of each subdivided edge will have very small weight (so that it is never beneficial to remove it from the tour), whereas the other half has weight equal to the weight of the original undivided edge. This will ensure that an optimal \kmove in the resulting instance removes only disjoint edges from the tour.

Formally, the instance~$(G',T',d')$ is produced as follows. Graph~$G'$ consists of vertices $a_1, b_1, \linebreak[1] a_2, b_2, \ldots, a_n, b_n$ and the initial tour~$T'$ visits the vertices in this order. The (symmetric) distances $d'(\cdot,\cdot)$ between these vertices are defined as follows:
\begin{itemize}
	\item $d'(a_i,a_j)= d'(b_i, b_j) = d(v_i,v_j)$ for~$1 \leq i, j \leq n$;
	\item $d'(a_i,b_i)=-2kM$ for $1\le i\le n$;
	\item $d'(a_i,b_j)=d(v_i,v_j)$ for~$1 \leq i,j \leq n$ with~$i \neq j$.
\end{itemize}

\begin{figure}
\begin{center}
\includegraphics{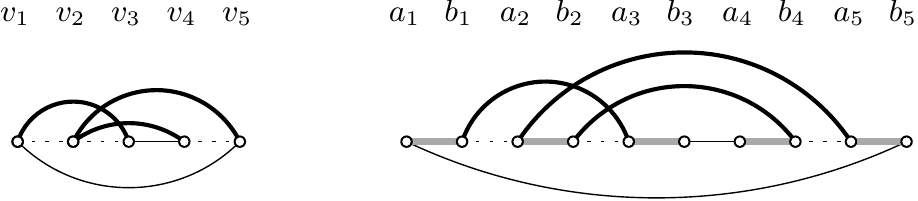}
\end{center}
\caption{Illustration for the reduction of Lemma~\ref{lemma:subdividing}. Left: illustration of an instance~$(G,T,d)$ with five vertices where~$T = (v_1, v_2, v_3, v_4, v_5)$. The \Threemove that removes edges~$E = \{ \{v_1, v_2\}, \{v_2, v_3\}, \{v_4, v_5\} \}$ is shown. Removed edges~$E$ are dotted, inserted edges~$F$ are thick. Right: the instance~$(G',T',d')$ resulting from the reduction by subdividing edges. The edges with very small weight are gray; they cannot be removed by an improving \kmove. The \Threemove for~$T$ has a natural analogue in~$T'$, where the removed edges share no endpoints.}
\label{fi:subdividing}
\end{figure}

The first two properties of the reduction follow immediately from these definitions. Let us consider how a \kmove that improves~$T$ by removing the edges~$E$ and adding the edges~$F$, translates into a \kmove improving~$T'$. We may assume without loss of generality that~$E \cap F = \emptyset$, since~$E \Delta F$ is also a valid \kmove with the same effect. Removing the edges~$E$ from~$T$ splits the tour~$T$ into~$|E|$ paths. For each edge~$\{v_i,v_{i+1}\} \in E$ (modulo~$n$), we remove the edge~$\{b_i,a_{i+1}\}$ from~$T'$ to split it into~$|E|$ paths; note that~$d'(b_i, a_{i+1}) = d(v_i, v_{i+1})$. Let~$E'$ denote the corresponding set of removed edges of~$T'$. Every edge in~$F$ connects two endpoints of paths of~$T - E$. If~$v_i$ is an endpoint of a path in~$T - E$, then either~$a_i$ or~$b_i$ is an endpoint of a path in~$T' - E'$. (In the special case that both tour edges incident on~$v_i$ are in~$E$, there is a path in~$T' - E'$ consisting of only~$a_i$ and~$b_i$ and both vertices are endpoints.) For each inserted edge~$F$ between endpoints~$p,q$ of paths in~$T - E$, insert into~$F'$ the edge between the corresponding endpoints of the paths in~$T' - E'$, which has the same weight (see Figure~\ref{fi:subdividing}). It follows that replacing~$E'$ by~$F'$ changes the weight of tour~$T'$ in the same way as replacing~$E$ by~$F$ does for tour~$T$. Let~$\opt$ denote the optimal cost improvement achieved by a \kmove for~$T$, and let~$\opt'$ denote the optimal improvement for~$T'$. Applying this transformation to an optimal \kmove for~$T$ shows that $\opt' \geq \opt$ and yields a \kmove for~$T'$ with profit~$\opt$ for which the removed edges share no endpoints. To prove the third property, it suffices to show that~$\opt' \leq \opt$, implying that such a move is also optimal for~$T'$. This will be implied by our proof of the fourth property, which we now present.

Consider a tour~$T'_1$ in~$G'$ obtained by applying an optimal \kmove to~$T'$. We claim that tour~$T'_1$ contains all edges~$\{a_i,b_i\}$ for~$1 \leq i \leq n$. To see this, observe that those are the only edges of weight~$- 2kM$, and the other edges have weight at least~$-M$. If one of these edges~$\{a_i,b_i\}$ disappears from the tour (increasing its weight by~$2kM$), then at best the other~$k-1$ removed edges of weight~$\leq M$ decrease the weight by~$(k-1)M$, causing a net weight increase of~$(k+1)M$ due to removals. Inserting~$k$ distinct new edges into the tour decreases the weight by at most~$kM$, since there are no new edges of weight~$-2kM$ to introduce and the smallest weight under~$d$ is at least~$-M$. Hence any \kmove that removes an edge of weight~$-2kM$ is not optimal since it increases the weight of the tour; the empty \kmove that performs no changes is better. It follows that~$T'_1$ contains the~$n$ edges~$\{a_i,b_i\}$ for~$1 \leq i \leq n$. Since vertices~$a_i$ and~$b_i$ have the same distances to the rest of the vertices for all~$i$, in this tour~$T'_1$ we can ``contract'' all edges~$\{a_i,b_i\}$ for~$1 \leq i \leq n$ to obtain a tour~$T_1$ in~$G$ whose cost difference with~$T$ is the same as the difference between~$T'_1$ and~$T$, and which can also be obtained by a \kmove. It follows that the optimal weight improvement by applying a \kmove to~$T'$ is bounded by the optimal weight improvement by applying a \kmove to~$T$, showing that~$\opt' \leq \opt$ (and therefore~$\opt = \opt'$) and proving the third property. The \kmove used to obtain~$T_1$ can easily be extracted from the \kmove used to obtain~$T'_1$. Any removed or inserted edge in the \kmove producing~$T'_1$ connects two vertices with distinct indices~$i,j$ in the range~$1 \ldots n$; the \kmove to produce~$T_1$ removes or inserts the corresponding edge~$\{v_i,v_j\}$. This completes the proof of the fourth property.
\end{proof}


\section{Faster 2-OPT: Additional details}

\subsection{The repeated case}

In this section we provide the proof of Theorem~\ref{theorem:repeated:twoopt}.

{
\renewcommand{\thetheorem}{\ref{theorem:repeated:twoopt}}
\begin{theorem}
After $\Oh(n^2)$ preprocessing and using $\Oh(n^2)$ storage we can repeatedly solve
the \TwoOptOPT problem in $\Oh(n\log n)$ time per iteration.
\end{theorem}
\addtocounter{theorem}{-1}
}

\begin{proof}
Let $T$ be the current tour, which is either the initial tour or the tour resulting from
the previous operation. Note that a \Twomove not only replaces a pair of edges by another
pair, but that it also reverses the subpath connecting these edges. To avoid spending time
on each edge of the subpath when we perform a \Twomove, we borrow an idea from
Chrobak~\etal~\cite{ChrobakSK99} that was also used by Fredman~\etal~\cite{FredmanJMO95}:
we store the tour in a tree, and with each node~$\node$
we store a Boolean $\rev(\node)$ indicating whether the subpath represented by the
subtree $\tree_\node$ rooted at $\node$ should be reversed.
(Fredman~\etal~\cite{FredmanJMO95} also use this idea to speed up \TwoOPT. However, their goal is only to be able to perform a \Twomove efficiently, and so they only maintain one such tree for the whole tour. Our goal is to find a \Twomove efficiently.)
In fact (and unlike Fredman~\etal) we will maintain $n$ such trees---for
each edge $e$ in the tour we maintain a tree $\tree(e)$ on the path
$\mypath(e) := T\setminus \{e\}$---and we augment these trees with extra information,
so that we can quickly find the best edge for $e$ to perform a \Twomove with.
The tree $\tree(e)$ is defined as follows.
\medskip

Fix an arbitrary orientation for $e$. This induces an orientation on the tour~$T$
and, hence, on the path $\mypath(e)$. The tree $\tree(e)$ is a red-black tree storing
the edges from $\mypath(e)$ in its leaves and storing a Boolean
$\rev(\node)$ at each node~$\node$. Initially
the order of the edges corresponds to the order along $\mypath(e)$
and all Booleans $\rev(\node)$ are set to $\false$. Later the order of
the edges along $\mypath(e)$ may no longer correspond to the
order of the leaves, but the correct order can always be restored
by ``pushing down'' the reversals in a top-down manner.
(To push down a reversal for a node $\node$ with $\rev(\node)=\true$ we
swap the left and right subtree of $\node$, set $\rev(\node)$ to
$\false$, and negate the Booleans $\rev(\cdot)$ of the children of $\node$.
This operation is called \emph{clearing} the node by Fredman~\etal
Note that swapping two subtrees of a node does not influence the red-black properties.)
So far our tree is essentially the same as that of Chrobak~\etal and Fredman~\etal
We now augment $\tree(e)$ as follows.

Let the \emph{local orientation} of an edge $e'$ in $\mypath(e)$ at the leaf~$\node$ where it is stored
be defined as follows: if $\rev(\node)=\false$ then the local orientation is the orientation along $\mypath(e)$
when $\tree(e)$ was constructed (that is, before any reversals took place), otherwise it is the opposite orientation.
The local orientation of $e'$ at an internal node $\node$ with
$e'$ in its subtree is defined recursively: if $\rev(\node)=\false$
then the local orientation of $e'$ at $\node$ is equal to the
local orientation at the relevant child of~$\node$, and if
$\rev(\node)=\true$ then it is the reverse of that orientation.
Note that the local orientation of $e'$ at the root of $\tree(e)$
is equal to the current orientation of $e'$ in $\mypath(e)$.
We store the following extra information at each node $\node$:
\begin{itemize}
\item A value $\mc(\node)$, which is defined as the minimum
      over all edges $e'$ in $\tree_\node$ of the cost of the \Twomove defined
      by $e$ and $e'$ for the local orientation of $e'$ at $\node$.
      We also store a pointer to the edge $e'$ defining the minimum.
\item A value $\mrc(\node)$ (with the corresponding pointer)
      which is defined similarly as $\mc(\node)$,
      except that we consider the reverse of the local orientations.
\end{itemize}
Note that if $\node_1$ and $\node_2$ are the two children of $\node$ then
\begin{equation} \label{eq:tree}
\mc(\node) \ = \ \left\{ \begin{array}{ll}
                    \min (\mc(\node_1),\mc(\node_2))  & \mbox{if $\rev(\node)=\false$} \\
                    \min (\mrc(\node_1),\mrc(\node_2)) & \mbox{if $\rev(\node)=\true$}.
                    \end{array}
             \right.
\end{equation}
Similarly, $\mrc(\node)$ can be computed in $\Oh(1)$ time from the information
at $\node$, $\node_1$, and~$\node_2$. Note that when $\rev(\node)$ is negated,
we can just swap the values of $\mc(\node)$ and $\mrc(\node)$
and propagate the change upward.
For each edge $e'$ we also maintain, for each tree~$\tree(e)$,
a pointer to the leaf where $e'$ is stored.
Next we show how to use the trees $\tree(e)$ to perform a \TwoOPT iteration in near-linear time.
\medskip

Finding the best \Twomove in $\Oh(n)$ time is easy: we simply go over all trees $\tree(e)$ to find the
one minimizing $\mc(\myroot(\tree(e))$. Let $e'$ be the edge
defining this value. We now have to perform a \Twomove on $e,e'$ (assuming $\mc(\myroot(\tree(e))$
is negative, that is, that the \Twomove actually reduces the cost of the tour).
Performing the \Twomove is done as follows. We first walk from the leaf storing
$e'$ back up to the root, to determine the current orientation of~$e'$. With that information we
can compute the edges $\tilde{e}$ and $\tilde{e}'$ that should replace $e$ and $e'$.
We destroy the trees $\tree(e)$ and $\tree(e')$, and build trees $\tree(\tilde{e})$
and $\tree(\tilde{e}')$ from scratch. The latter can be done in $\Oh(n)$ time
after constructing the path $\mypath(\tilde{e})$ and $\mypath(\tilde{e}')$,
which we can do in $\Oh(n)$ time. It remains to update the other trees. In the remainder of the proof we show how this can
be done in $\Oh(\log n)$ time per tree, resulting in $\Oh(n\log n)$ time in total for a \Twomove.

We show how to update a tree $\tree(f)$ in logarithmic time
when a \Twomove with edges $e,e'$ is performed; see also Fig.~\ref{fi:PathTree}.
\begin{figure}
\begin{center}
\includegraphics{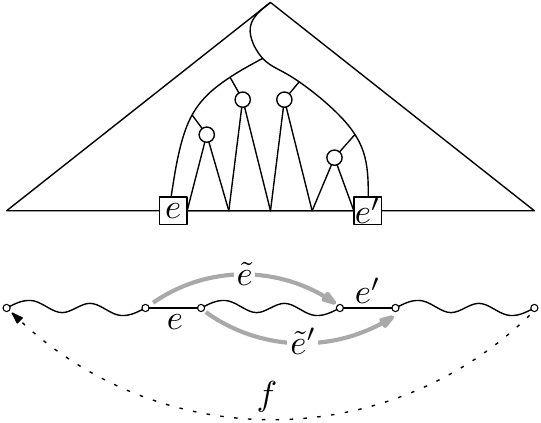}
\end{center}
\caption{Situation when a \Twomove with edges $e,e'$ has to be performed on a tree $\tree(f)$.
         When all nodes on the search paths to $e$ and $e'$ are cleared, the subtrees in between
         the search paths together represent the subpath from $e$ to $e'$ along $\mypath(f)$.}
\label{fi:PathTree}
\end{figure}
Note that rotations in~$\tree(f)$ can still be done in $\Oh(1)$ time after clearing
the two nodes on which the rotation is performed. Thus standard operations on
augmented red-black trees can still be performed in logarithmic
time. These operations include insertions and deletions, but also splits and
concatenations. In a split operation in a normal red-black tree one is given
a value~$X$, and the goal is to split the tree into two new trees: one containing
the elements smaller than $X$, and one containing the elements larger than~$X$.
We will need to split $\tree(f)$, given an edge $e\in \mypath(f)$, into two
trees: one for the part of $\mypath(f)$ before~$e$, and one for the part starting at~$e$.
This is possible in the usual way, provided we first clear all nodes on the path
from the root of $\tree(f)$ to the leaf containing~$e$. Similarly, concatenating
two trees---the reverse operating from splitting---can be done in $\Oh(\log n)$ time.
See also the paper by Chrobak~\etal~\cite{ChrobakSK99}, who describe these operations
(for AVL-trees) and without the extra fields $\mc(\cdot)$ and $\mrc(\cdot)$.
We can now update $\tree(f)$ (to reflect a \Twomove where edges $e$ and $e'$ are replaced
by new edges $\tilde{e}$ and $\tilde{e}'$) as follows.

We first split $\tree(f)$ into two subtrees, a tree $\tree_1$ for the
subpath of $\mypath(f)$ before $e$, and a tree $\tree_2$ for the subpath
starting at~$e$. The latter tree is then split further into a tree $\tree_{2,1}$
for the subpath from $e$ to $e'$, and a tree~$\tree_{2,2}$ for the
subpath behind~$e'$.
We then delete  $e$ and $e'$ from $\tree_{2,1}$, reverse the subpath in between them
by negating the Boolean $\rev(\cdot)$ at the root of $\tree_{2,1}$, and
insert $\tilde{e}$ as first edge of the subpath and $\tilde{e}'$ as last edge.
We then concatenate the three subtrees again to obtain the new tree $\tree(f)$.

We conclude that each tree $\tree(f)$ can be updated in $\Oh(\log n)$
time after a \Twomove.
\end{proof}

\subsection{Repeated 3-OPT}
\label{subse:appendix-repeated-3opt}
The approach above can also be used to speed up \ThreeOPT computations in the repeated
setting. To this end we maintain a data structure $\ds(e,e')$ for each pair $e,e'$
of edges in the tour, which allows us to quickly find the edge $e''$ that gives
the best \Threemove with $e,e'$. This data structure, which is very similar to the one
for \TwoOPT, is defined as follows.
\medskip

Let $e,e'$ be a pair of edges from the current tour.
To define $\ds(e,e')$ it is convenient to consider the moment at which $\ds(e,e')$ was created,
which is the first moment $e$ and $e'$ both appear in the tour. Let $T_{\mathrm{init}}$
be this initial tour (for the pair $e,e'$).
Fix an orientation for the edge~$e$. This orientation determines the orientation
of all other edges in $T_{\mathrm{init}}$, including the edge~$e'$. We call this orientation of~$e'$
its \emph{initial orientation}. Note that when the tour changes due to a \Threemove,
the orientation of $e'$ may change: instead of having the oriented edges $e,e'$ in the tour
we may have $e,\myrev(e')$ in the tour, where $\myrev(e')$ is the reverse of the edge~$e'$.
To deal with this, we will actually maintain two data structures for the pair of (undirected)
edges: a tree $\tree(e,e')$ and a tree $\tree(e,\myrev(e'))$.
Moreover, we will maintain a Boolean indicating whether the current tour uses~$e'$ or~$\myrev(e')$.
It is easily checked that these Booleans can be maintained without affecting the overall time bound.

Now consider $\tree(e,e')$; the tree $\tree(e,\myrev(e'))$ is similar. Note that the orientations
of $e$ and~$e'$ are fixed. Hence, if we consider a third edge $e''$ in the current tour~$T$ and we know
the orientation of $e''$, then we also know which are the valid \Threemoves for the triple~$e,e',e''$;
see Fig.~\ref{fi:3-change}.
\begin{figure}
\begin{center}
\makebox[\linewidth][c]{
\includegraphics{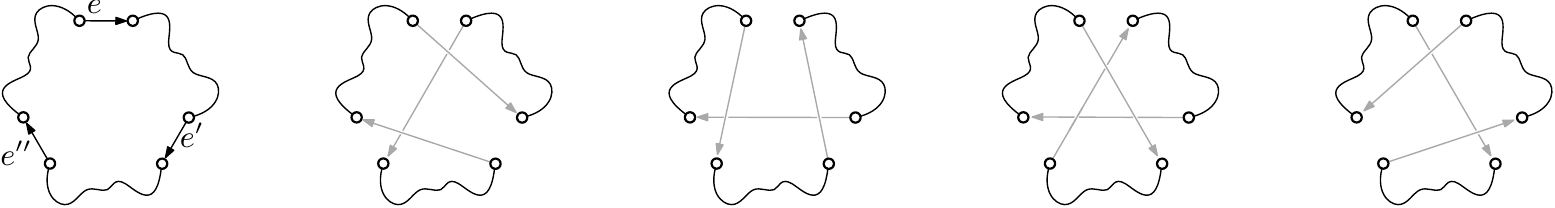}
}
\end{center}
\caption{A triple $e,e',e''$ in the current tour (left) has at most four different \Threemoves.
(``At most'' because if one or more of the subpaths between these edges are empty,
then some of these \Threemoves will degenerate into \Twomoves.)}
\label{fi:3-change}
\end{figure}
Thus we can define the tree $\tree(e,e')$ in a similar fashion we defined the tree $\tree(e)$
in the \TwoOPT setting. To this end, let $\mypath(e) := T\setminus \{e\}$ be the path resulting
from removing~$e$ from the tour~$T$. 
As before, we store $\mypath(e)$ in a red-black tree with Booleans $\rev(\node)$ at each node~$\node$,
and we define the local orientation of an edge $e''$ at a given node $\node$, where $\node$
must be such that $e''$ is stored in the subtree rooted at~$\node$.
We augment the nodes of $\tree(e,e')$ with the following extra information.
\begin{itemize}
\item A value $\mc(\node)$, which is defined as the minimum
      over all edges $e''\neq e'$ in $\tree_\node$ of the minimum cost of a valid \Threemove defined
      by $e,e',e'''$ for the local orientation of $e'$ at $\node$ and the fixed orientations of~$e$ and~$e'$.
      (As mentioned, for each $e''$ there are up to four types of valid \Threemoves.)
      We also store a pointer to the edge $e''$ defining the minimum
      and the type of the \Threemove.
\item A value $\mrc(\node)$ (with the corresponding pointer and type)
      which is defined similarly as $\mc(\node)$,
      except that we consider the reverse of the local orientations.
\end{itemize}
Note that the recurrence relation~(\ref{eq:tree}) still holds.

An iteration of the repeated \ThreeOPT algorithm now proceeds as follows.
For each pair of edges~$e,e'$ from the current tour we find the cheapest \Threemove
involving $e,e'$ by considering the relevant tree---either
$\tree(e,e')$ or $\tree(e,\myrev(e'))$, depending on the orientation of $e'$
in the current tour. Recall that the cheapest \Threemove is stored at
the root of the tree, so it can be found in $O(1)$ time. This gives
us $O(n^2)$ \Threemoves to consider---one for each pair~$e,e'$.
If the best of these \Threemoves has negative cost, 
we perform that \Threemove and update our data structures. The latter is
done as follows.

Let $e,e',e''$ be the old edges in the tour and let~$\tilde{e},\tilde{e}',\tilde{e}''$
be the new edges that replace them. We first destroy all data structures defined by any
of the old edges~$e,e',e''$, that is, all data structures $\ds(f,f')$ such that
$\{f,f'\} \cap \{e,e',e''\} \neq \emptyset$. Next we build (from scratch) 
all data structures defined by the new edges~$\tilde{e},\tilde{e}',\tilde{e}''$ . 
Since any edge is involved in $n-1$ pairs, the total number of data structures we
destroy and create is~$O(n)$. Since building a data structure can be done in
$O(n)$ time, this takes $O(n^2)$ time in total.

Ir remains to update the trees $\tree(f,f')$ and
$\tree(f,\myrev(f'))$ for $\{f,f'\} \cap \{e,e',e''\} = \emptyset$.
This can be done similarly to the \TwoOPT case.
More precisely, to update $\tree(f,f')$ we proceed as follows. First we split $\tree(f,f')$
into four subtrees by deleting the edges $e,e',e''$. Then we reverse one or
more of the resulting subpaths by negating the Boolean $\rev(\cdot)$ at the root 
of the corresponding subtree; which paths have to be reversed depends on the
specific type of \Threemove we have to perform. Finally, we insert the
new edges~$\tilde{e},\tilde{e}',\tilde{e}''$ into the relevant subtrees,
and we concatenate all subtrees to form the new tree $\tree(f,f')$.
Thus updating the tree$\tree(f,f')$---and, similarly,
updating $\tree(f,\myrev(f'))$---takes $O(\log n)$ time. Since we have to do
this for $O(n^2)$ pairs $f,f'$ we spend $O(n^2\log n)$ time in total.

\subsection{The planar case}
\label{subse:appendix-planar-case}
We now turn our attention to the planar setting. (Note that we do not consider the repeated
version of the problem, but the single-shot version.)
We focus on the problem of detecting any
\Twomove or \Threemove that lowers the cost of the tour, although similar results are possible
for the finding the best change.
{
\renewcommand{\thetheorem}{\ref{thm:twodet:threedet:planar}}
\begin{theorem}
For any fixed $\eps>0$,
\TwoOptDET in the plane can be solved in $\Oh(n^{8/5+\eps})$ time,
and \ThreeOptDET in the plane can be solved in $\Oh(n^{80/31+\eps})$ expected time.
\end{theorem}
\addtocounter{theorem}{-1}
}

\paragraph{2-OPT.}
Suppose we are given a tour $T$ on a planar point set $P := \{p_0,p_1,\ldots,p_{n-1}\}$,
where we assume without loss of generality that the points are numbered in order along~$T$.
The idea is to preprocess $P$ such that we can answer the following queries:
given a query edge $p_i p_{i+1}$ of $T$, find an edge $p_j p_{j+1}$ in $T$
such that performing a \Twomove on $e,e'$ lowers the cost of~$T$ (if such an edge exists).
In other words, we want to find an edge $p_j p_{j+1}$ such that
\begin{equation}
|p_i p_j| + |p_{i+1} p_{j+1}| <  |p_i p_{i+1}| + |p_j p_{j+1}|.   \label{eq:2change}
\end{equation}
To answer these queries we map every edge $p_j p_{j+1}$ to a point
$q_j := (x(p_j),y(p_j),x(p_{j+1}),y(p_{j+1}))$ in $\Reals^4$,
and we preprocess the resulting set of points in $\Reals^4$ for range queries
with semi-algebraic sets~\cite{ams-rssasII-13}. Given a query edge $p_i p_{i+1}$
we define a range $Q_i\subset \Reals^4$ as
\[
Q_i := \ \{ \ (a_1,a_2,a_3,a_4) : \ |p_i \; (a_1,a_2)| + |p_{i+1} \; (a_3,a_4)| <  |p_i p_{i+1}| + |(a_1,a_2) (a_3,a_4)| \ \}.
\]
Thus $p_j p_{j+1}$ satisfies (\ref{eq:2change}) if and only if $q_j \in Q_i$.
We can therefore find an edge $p_j p_{j+1}$
satisfying (\ref{eq:2change}) by performing a query with the range~$Q_i$,
which is a semi-algebraic set.
In $\Reals^4$, semi-algebraic range-searching queries can be answered in
$\Oh(n^{3/4+\eps})$ time after $\Oh(n^{1+\eps})$ preprocessing~\cite{am-rssas-94,k-atuvd-04}.
Alternatively, we can ``dualize'' the approach, by mapping each edge $p_j p_{j+1}$
to a surface in $\Reals^4$ and mapping the query $p_i p_{i+1}$ to a point~$q_i$.
By performing point location with $q_i$ in the arrangement defined by these surfaces
we can then answer the queries. This takes $\Oh(\log n)$ time after $\Oh(n^{4+\eps})$
preprocessing~\cite{k-atuvd-04}. By combining these two solutions in a standard manner, we can obtain a
trade-off between preprocessing and query time---see e.g.~\cite{m-rsehc-93} and also below,
where we give some more details for the somewhat more complicated case of \ThreeOPT.
In particular, we can obtain
$\Oh(n^{3/5+\eps})$ query time after $\Oh(n^{8/5+\eps})$ preprocessing.
Thus our \TwoOPT algorithm needs $\Oh(n^{8/5+\eps})$ time in total.

\paragraph{3-OPT.}
For \ThreeOPT we proceed similarly as for \TwoOPT. We preprocess the tour $T$ for the following
queries: given a query edge~$p_i p_{i+1}$, find a pair of edges
$p_j p_{j+1}, p_k p_{k+1}$ such that a \Threemove involving these three edges will reduce
the cost of the tour (if such a pair exists). The details are a bit more involved
than for \TwoOPT, however.

Assume the points are numbered $p_0,\ldots,p_{n-1}$ in order along~$T$.
Define $e_i$ to be the edge $v_i v_{i+1}$, for $0\leq i<n$, and consider a \Threemove
involving edges $e_i,e_j,e_k$ with $i<j<k$. The four possible triples to replace
$e_i,e_j,e_k$ in a valid \Threemove are
\myitemize{
\item $p_i p_j, p_{i+1}p_k, p_{j+1} p_{k+1}$ (Type~I);
\item $p_i p_{j+1}, p_k p_{i+1}, p_j p_{k+1}$ (Type~II);
\item $p_i p_{j+1}, p_k p_j, p_{i+1} p_{k+1}$ (Type~III);
\item $p_i p_k, p_{j+1}p_{i+1}, p_j p_{k+1}$ (Type~IV).
}
Here we have ignored the possibility that one of the edges $e_1,e_2,e_3$
re-appears in the new triple, and we thus have a \Twomove; these ``degenerate''
\Threemoves can be found as described above.
Note that we may have $i+1=j$ and/or $j+1=k$. In this case some of the
four \Threemoves just mentioned also become degenerate, but this is not a problem.
Indeed, these \Threemoves still result in a valid tour, and if the tour length
is reduced we still want to find such a degenerate \Threemove.
We are left with the problem of deciding whether there is a \Threemove
of one of the four types described above that reduces
the length of the tour. We explain how to do this for \Threemoves of Type~I;
the other three types can be handled similarly.
\medskip

To find an improving \Threemove of Type~I we store all pairs $e_i,e_j$ with $0\leq i<j<n-1$
in a data structure that can answer the following queries: given an edge $e_k$, find
a pair $e_i,e_j$ such that
\begin{equation}
j<k \ \mbox{ and } \ |p_i p_j| + |p_{i+1}p_k| + |p_{j+1} p_{k+1}| < |p_i p_{i+1}| + |p_j p_{j+1}| + |p_k p_{k+1}|,
   \label{eq:3change}
\end{equation}
if such a pair exists.
For the moment, let's ignore the condition~$j<k$. Then we can proceed similarly as in
the \TwoOPT case: we map every pair $e_i,e_j$ to a point
\[
q_{ij} := \left( x(p_i), y(p_i), x(p_{i+1}), y(p_{i+1}), x(p_j), y(p_j), x(p_{j+1}), y(p_{j+1})  \right)
\]
in $\Reals^8$, and we preprocess the resulting set of points for range queries
with semi-algebraic sets~\cite{ams-rssasII-13}.
Given a query edge $p_k p_{k+1}$ we can now decide if there is an improving \Threemove of Type~I by searching with the range
\[
\begin{array}{llll}
Q_k & := & \{ \ (a_1,\ldots,a_8) : & |(a_1,a_2) (a_5,a_6)| + |(a_3,a_4) p_k| + |(a_7,a_8) p_{k+1}|  \\
    &    &                         & < |(a_1,a_2) (a_3,a_4)| + |(a_5,a_6) (a_7,a_8)| + |p_k p_{k+1}| \}.
\end{array}
\]
The resulting data structure uses $\Oh(n')$ space, and
has $\Oh((n')^{1+\eps})$ expected preprocessing time and $\Oh((n')^{7/8+\eps})$ query time,
where $n'$ is the number of points stored in the data structure.

Alternatively, we can map every pair $e_i,e_j$ to a surface
\[
\begin{array}{llll}
\Gamma_{ij} & := &  \{ (a_1,\ldots,a_4) : & |p_i p_j| + |p_{i+1}(a_1,a_2)| + |p_{j+1} (a_3,a_4)| \\
            &    &                        & = |p_i p_{i+1}| + |p_j p_{j+1}| + |(a_1,a_2) (a_3,a_4)| \}
\end{array}
\]
in~$\Reals^4$, and preprocess the resulting arrangement for point location.
Performing a point-location query with the point~$(x(p_k),y(p_k),x(p_{k+1}),y(p_{k+1}))$
now tells us if there is an improving \Threemove of Type~I.
This alternative would use $\Oh((n')^{4+\eps})$
preprocessing time and have $\Oh(\log n')$ query time~\cite{k-atuvd-04}.
\medskip

The standard way to obtain a trade-off between preprocessing and query time is
as follows. The linear-space variant is a recursively defined tree structure
on the points in the input set (which is in our case the set $\{ q_{ij} : 0\leq i<j<n-1 \}$).
Now, instead of continuing the recursion all the way until only constantly many points
are left, we stop when the number of points falls below a suitable threshold~$m$ with $1\leq m\leq n'$.
(The value of $m$ determines the trade-off.)
At this point we dualize the problem and build the logarithmic query-time solution,
which in our case uses $\Oh(m^{4+\eps})$ preprocessing time.
This way we construct a ``top tree'' with $\Oh(n'/m)$ leaves, each of which is associated
with a ``bottom tree'' that needs $\Oh(m^{4+\eps})$ preprocessing.
The total amount of preprocessing is $\Oh(n' m^{3+\eps})$.

A query is performed by first searching in the top tree. The search ends up in $\Oh((n/m)^{7/8+\eps})$
leaves where the search is then continued in the corresponding bottom tree.
Thus the query time is $\Oh((n'/m)^{7/8+\eps})$ (for a slightly larger $\eps$, which
swallows the extra log-factor from searching in the bottom trees).
\medskip

So far we ignored the condition~$j<k$ in~(\ref{eq:3change}).
Fortunately this condition is easy to handle, as it simply adds a so-called
\emph{range restriction} to the query. Range restrictions can be
added at the cost of an extra log-factor in preprocessing time and query time~\cite{wl-arrdd-85}.
In our case these logarithmic factors are swallowed by the $\Oh(n^{\eps})$ factor that we already have,
so the total structure uses  $\Oh(n' m^{3+\eps})$ expected preprocessing time and has $\Oh((n'/m)^{7/8+\eps})$
query time, where $m$ is a parameter
that we can still change to optimize performance.

Recall that $n'$, the number of points stored in the data structure, is~$n^2$,
and that we have to perform $n$ queries---one for each edge $e_k$. Thus the total
time of our algorithm is
\[
\Oh(n^2 m^{3+\eps}) + n \cdot \Oh((n^2/m)^{7/8+\eps}) = \Oh(n^2 m^{3+\eps} + n^{22/8+\eps} / m^{7/8}).
\]
This is minimized when we set $m := n^{6/31}$, which gives a total expected runtime of
$\Oh(n^{80/31+\eps}) = \Oh(n^{2.59})$.

\end{document}